\documentclass[11pt]{article}
\usepackage[a4paper]{geometry}
\usepackage[italian,english]{babel}

\usepackage{amsmath}
\usepackage{amsfonts}
\usepackage{amstext}
\usepackage{amssymb}
\usepackage{amsthm}
\usepackage{amscd}

\usepackage{fancyhdr}
\usepackage{hyperref}
\hypersetup{colorlinks,
linkcolor=myrefcolor,
citecolor=mycitecolor,
urlcolor=myurlcolor}
\usepackage[capitalize]{cleveref}
\usepackage{caption}
\usepackage{etaremune}
\usepackage{bibentry}

\usepackage{xcolor}
\definecolor{myurlcolor}{rgb}{0,0,0.4}
\definecolor{mycitecolor}{rgb}{0,0.5,0}
\definecolor{myrefcolor}{rgb}{0.5,0,0}
\usepackage{graphicx}
\usepackage{tikz}
\usepackage{tikz-cd}
\usepackage{mathrsfs}

\usepackage{etoolbox}
\usepackage{makeidx}
\usepackage{sectsty}
\usepackage{dsfont}
\usepackage{enumitem} 
\usepackage[]{latexsym}
\usepackage{braket}
\usepackage{caption}
\usepackage[utf8]{inputenx}
\usepackage[T1]{fontenc}
\usepackage{lmodern}
\usepackage{textcomp}
\usepackage{microtype}
\usepackage{totcount}
\usepackage{blindtext}

\newtheorem{theorem}{Theorem}

\newtheorem{proposition}{Proposition}
\newtheorem{definition}{Definition}

\newtheorem*{proof*}{Proof}

\newcommand{\be}{\begin{equation}}
\newcommand{\ee}{\end{equation}}
\newcommand{\bea}{\begin{eqnarray}}
\newcommand{\eea}{\end{eqnarray}}

\newcommand{\blue}[1]{\textcolor{blue}{{#1}}}

\title{Schwinger's Picture of Quantum Mechanics IV: Composition and independence}

\author{F. M. Ciaglia$^{1,5}$  \href{https://orcid.org/0000-0002-8987-1181}{\includegraphics[scale=0.7]{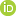}}, F. Di Cosmo$^{2,3,6}$ \href{https://orcid.org/0000-0003-0256-5913}{\includegraphics[scale=0.7]{ORCID.png}}, A. Ibort$^{2,3,7}$\href{https://orcid.org/0000-0002-0580-5858}{\includegraphics[scale=0.7]{ORCID.png}}, G. Marmo$^{4,8}$\href{https://orcid.org/0000-0003-2662-2193}{\includegraphics[scale=0.7]{ORCID.png}}\\
\footnotesize{$^{1}$\textit{ Max Planck Institute for Mathematics in the Sciences, Leipzig, Germany}} \\
\footnotesize{$^{2}$\textit{ ICMAT, Instituto de Ciencias Matem\'{a}ticas (CSIC-UAM-UC3M-UCM)}}  \\
\footnotesize{$^{3}$\textit{ Depto. de Matem\'aticas, Univ. Carlos III de Madrid, Legan\'es, Madrid, Spain}}  \\
\footnotesize{$^{4}$\textit{ Dipartimento di Fisica ``E. Pancini'', Universit\`a di Napoli Federico II, Napoli, Italy}} \\
\footnotesize{$^{5}$\textit{ e-mail: \texttt{florio.m.ciaglia[at]gmail.com}}}, \\
\footnotesize{ $^{6}$\textit{ e-mail: \texttt{fabio[at]}}} \\ 
\footnotesize{ $^{7}$\textit{ e-mail: \texttt{albertoi[at]math.uc3m.es}}} \\ 
\footnotesize{$^{8}$\textit{ e-mail: \texttt{marmo[at]na.infn.it}}}}

\begin{document}

\maketitle

\begin{abstract}
The groupoids description of Schwinger's picture of quantum mechanics is continued by discussing the closely related notions of composition of systems, subsystems, and their independence.  Physical subsystems have a neat algebraic description as subgroupoids of the Schwinger's groupoid of the system.  The groupoids picture offers two natural notions of composition of systems: Direct and free products of groupoids, that will be analyzed in depth as well as their universal character.  Finally, the notion of independence of subsystems will be reviewed, finding that the usual notion of independence, as well as the notion of free independence, find a natural realm in the groupoids formalism.   The ideas described in this paper will be illustrated by using the EPRB experiment.  It will be observed that, in addition to the notion of the non-separability provided by the entangled state of the system, there is an intrinsic `non-separability' associated to the impossibility of identifying the entangled particles as subsystems of the total system.
\end{abstract}

\tableofcontents

\thispagestyle{fancy}

\section{Introduction} 

\subsection{The groupoids picture of Quantum Mechanics}
In the series of papers \cite{Ci18,Ib18a, Ib18b, Ib18c} some relevant traits of a new picture of Quantum Mechanics emerging from Schwinger's algebra of selective measurements \cite{Sc91, Sc01} have been discussed.  

In \cite{Ib18a} the kinematical background of the theory was presented as an abstraction of Schwinger's algebra of selective measurements.  It was shown that the categorical notion of 2-groupoids provides a natural mathematical background for the symbolic language needed to describe microscopic physical processes as sought by Schwinger.  

The notions of states and observables, together with the important notion of amplitudes, were started to be discussed in \cite{Ib18b} as well as a first Hamiltonian approach towards a dynamical setting of quantum theories.  In \cite{Ib18c} the statistical interpretation of quantum mechanical systems was addressed.  It was shown that there is a perfect match between the statistical interpretation of Quantum Mechanics provided by R. Sorkin's quantum measures and decoherence functionals and the notion of states in the algebra of the groupoid describing the theory.   

In the present paper, the analysis of the foundations of quantum mechanics is pursued further by addressing the relation of the theoretical background provided by the groupoids interpretation of J. Schwinger's algebra of selective measurements and the problem of composition of systems, which is of fundamental importance in the interpretation of Quantum Mechanics, as well as the intimately associated notions of subsystem and independence, two other key notions to develop a proper theory of information and thermodynamics of quantum systems.    


\subsection{The problem of composition}

The composition of systems is a fundamental notion in our description of the physical world.    As it was expressed by E. Lieb and J. Yngvason \cite{Li00, Li99}, on its most common/simple/naive acceptation, composition of systems \textit{``corresponds simply to the two (or more) systems being side by side on the laboratory table; mathematically it is just another system called a \textit{compound system}''}.  This operational notion of composition fits well in the groupoid picture as it simply states that whatever will be the meaning of the expression `systems being side by side on the laboratory table', there is a natural interpretation of such situation corresponding to a experimental setting where the outcomes and transitions of both systems can be observed simultaneously and independently of each other.    As it will be discussed later on in this paper, this notion will correspond to the direct product of the corresponding groupoids.

The rightness and physical feasibility of this way of proceeding, that barely deserves a nod in classical physics, becomes acutely relevant in quantum mechanics.  Just to mention two foundational instances: Von Neumann introduced the notion of composition of system to provide a consistent description of the measuring process \cite[Chap. VI, Sects. 1,2, pag. 224 ff.]{Ne32}.  Later on, in the epoch-making analysis of the completeness of quantum mechanics, Einstein, Podolski and Rosen proposed and experiment where the role of composition of subsystems is critical \cite{Ei35}.

Recent developments in quantum information have led to consider general abstractions of such composition procedures and, from a categorical perspective \cite{Ab08}, the mathematical structure underlying the notion of composition is the structure of a monoidal category \cite{Se10}.  In this setting the composition of systems is described as an abstract binary operation, that can be denoted as $\otimes$ (even if, in general, it is not the standard tensor product of algebras or vector spaces).  The system $A \otimes B$ is interpreted as the composite system made of subsystems $A$ and $B$.  Thus systems are constructed by combining subsystems together.   

The groupoid picture of quantum mechanics is categorical in a natural way and, in these sense, the previous abstractions of the notion of composition could be used, however we will choose not to do that.  Instead, using  the conceptual framework provided by the groupoids description of quantum systems, we will proceed by analyzing first  the basic physical ideas behind the notion of subsystems, that appears as a natural ingredient in the notion of composition of systems, and we will proceed afterwards to address the problem of composition.  

Proceeding in this way, we will arrive   to the natural definition of subsystems as subgroupoids of the groupoid describing the given quantum system (see Sect. \ref{sec:subsystems} below).   Once the notion of subsystem is elucidated, we  simply have to consider the composition of two systems (i.e., groupoids) as any groupoid constructed out of them and such that possesses both as subgroupoids.  We will be able to analyze from this perspective two such constructions, the so called direct products and free products.   Both notions will be discussed at length in Sect. \ref{sec:composition}.   

As it will be shown in the main text, the direct product of groupoids correspond to the common definition of composition of systems, however it will be shown that this construction is not as innocent as was suggested before. In fact, taking a strict approach to the problem of composition, we will see that the direct product does not satisfy the conditions determining a proper composition because the original groupoids are not subgroupoids of their direct product.   This apparently paradoxical behaviour of the direct product could be responsible for some of the difficulties in the interpretation of experiments, like the Einstein-Podolski-Rosen-Bohm (EPRB) experiment, where it is a crucial ingredient.   

On the other hand, the notion of free product of groupoids is a proper composition operation where the original groupoids are natural subgroupoids of the free product.   The physical interpretation of this operation, that extends in a natural way the free product of groups, is closer to two relevant physical notions such as the histories interpretation of dynamics and the splitting and recombination of physical systems by introducing `walls'.   The latter ideas will be discussed in the main text and the former one will be left for subsequent developments. 

Deeply related to the notion of composition of subsystems is the notion of `independence' of two given subsystems.  From a physical perspective, we may follow Landau's  insight \cite[Chap. 1, p. 7]{La80}:

 ``\textit{By statistical independence we mean that the state of one subsystem does not affect the probabilities of various states of the other subsystems}.'' 

Determining conditions under which two given subsystems are independent, either statistically or physically, is critical to complete the analysis of composition of systems as it hinders the intepretation of the experiments performed on them.  The second part of this work will be devoted to the analysis of this problem both from a general theoretical perspective and taking into account the specifics introduced by the groupoids description of quantum systems.    

We will be able to provide a general notion of independence of subsystems that will encompass the usual notion of independence (which is modelled on the notion of tensor product of algebras, hence on the direct composition provided by the direct product of groipoids) as well as the extremely fruitful and genuinely non-commutative notion of free independence introduced by D.  Voiculescu \cite{Vo92, Vo94}.   This will be achieved by following the lead started in \cite{Ci19} where a first principle approximation to the notion of composition of systems was discussed.  It will be shown that, for appropriate selected states, both notions of independence, which are nevertherless different, reproduce Landau's notion of statistical independence of subsystems.    Such states for which the corresponding subsystems are statistically independent include, as it should be, separable states in the standard description of quantum mechanics, but many other relevant possibilities are open.  In particular it will be shown that an extension of the notion of factorizable states introduced in \cite{Ib18c} leads also to meaningful statistically independent subsystems.

This paper will end by applying the previous arguments to the particularly interesting situation posed by the EPRB argument.    The EPRB experiment will be analyzed from the point of view of groupoids.  It will be shown that the apparently paradoxical properties of the system are obtained without referring to any interpretation in terms of subsystems (separable or not) and, only if we insist on interpreting the system as a composition of two independent subsystems, such difficulties appear.  Actually, it will be shown that such  attitude is not consistent with the groupoid picture of quantum mechanical systems and that standard composition does not allow for a natural interpretation in terms of composition of independent subsystems of the original system.
 

\subsection{A succinct review of the groupoids picture of quantum mechanics}\label{sec:review}

We will provide a succinct review of the basic notions and terminology required to describe quantum systems in the groupoid picture that constitutes an abstraction of Schwinger's algebra of selective measurements (see \cite{Ib20} for a recent succinct review of the main notions and \cite{Ib18a, Ib18b, Ib18c} for more details).

In the groupoid picture of Quantum Mechanics, a physical system is described by a groupoid $\mathbf{G} \rightrightarrows \Omega$ where its morphisms \hbox{$\alpha\,:\,a\, \rightarrow \,a'$, $a,\,a' \in \Omega$} will be called \textit{transitions} and they are an abstraction of Schwinger's selective measurements or, even in more abstract terms, they provide an abstraction of the notion of `quantum jumps'. The objects of the groupoid, i.e., the elements $a \in \Omega$, will be called the \textit{outcomes} of the system and they correspond to the outcomes of observations performed on the system. Finally, the units corresponding to the outcome $a$ will be denoted by $\mathbf{1}_a$, $a\in \Omega$. The source and the target maps will be defined, respectively, as $s(\alpha)=a $  and $t(\alpha)=a'$, where $\alpha:\,a\,\rightarrow\,a'$ is the transition that changes the outcomes of the system when $\mathbb{A}$ is measured from $a$ to $a'$.  If the transitions $\alpha$, $\beta$ are such that $s(\alpha) = t(\beta)$, they will be said to be composable and their composition will be denoted by $\alpha \circ \beta$.

Some comments on the physical interpretation of the groupoid axioms are in order now. The constraints imposed on the composition of transitions and associativity property of the composition, $\alpha \circ (\beta \circ \gamma) = \alpha \circ (\beta \circ \gamma)$, $\forall \alpha, \beta, \gamma$ composable, reflect the causal structure of the experimental setting. Note that, in this sense, the right to left composition convention relates to the definition of the future in the laboratory system (while reading the composition from left to right will correspond to ``time'' flowing backwards in the given experimental setting).

The invertibility axiom, that is, $\forall \alpha, \exists \alpha^{-1}$ such that $\alpha \circ \alpha^{-1} = \mathbf{1}_{t(\alpha)}$ and $\alpha^{-1}\circ \alpha = \mathbf{1}_{s(\alpha)}$ reflects Feynman's principle of \textit{microscopic reversibility} (see \cite[p. 3]{Fe05}: ``\textit{The fundamental (microscopic) phenomena in nature are symmetrical with respect to the interchange of past and future}''). 

Thus, the physical interpretation of both entities, `outcomes' and `transitions', is given, using an operational approach, as all possible results obtained once a given experimental setting has been fixed by performing an arbitrary number of experiments on a family of compatible observables.  Transitions are, on the other side, the physical changes, either induced by the experimental setting or happening spontaneously, in the outcomes of the system. 
In other words, outcomes $a\in \Omega$ have the meaning of physical occurrences registered by the physical device used in studying the transitions of the system and they determine in a natural way the units of the groupoid, i.e., $\forall a \in \Omega$ there is a transition $\mathbf{1}_a$ that does not affect the system whenever it returns $a$ when $\mathbb{A}$ is measured, that is, such that $\alpha\circ \mathbf{1}_a = \alpha$; $\mathbf{1}_a\circ \alpha = \alpha$ $\forall \alpha :\, a\, \rightarrow \, a'$.

Given a groupoid $\mathbf{G} \rightrightarrows \Omega$ describing (so far just the kinematical structure of) a physical system, functions defined on it will be good candidates for observables of the system. Actually, we may construct in several ways a $C^*$-algebra of observables $C^*(\mathbf{G})$ associated with the groupoid. For instance, we may consider $C_c(\mathbf{G})$, the algebra of compactly supported functions in $\mathbf{G}$ equipped with the convolution product and $C^*(\mathbf{G}) = \overline{C_c(\mathbf{G})}$, the closure determined by some natural norm. For instance if $\mathbf{G}$ is discrete countable, $C_c(\mathbf{G})$ will be the set of functions $f:\,\mathbf{G}\, \rightarrow \, \mathbb{C}$ which are different from zero on a finite number of elements of $\mathbf{G}$, i.e., $f= \sum_{\alpha \in \mathbf{G}} f(\alpha)\delta_{\alpha}$, $f(\alpha)=0$ except for a finite number of $\alpha$'s, and the associative product is defined as
\begin{equation}
f \star g = \left( \sum_{\alpha \in \mathbf{G}} f(\alpha)\delta_{\alpha} \right) \star \left( \sum_{\beta \in \mathbf{G}} g(\beta)\delta_{\beta} \right) = \sum_{s(\alpha)=t(\beta)}f(\alpha)g(\beta)\delta_{\alpha\circ \beta}\,.
\end{equation}
Moreover, the adjoint $f^*$ of the function $f$, is defined as $f^* = \overline{f(\alpha)}\delta_{\alpha^{-1}}$ and the closure can be taken with respect to the weak topology on the space of bounded operator on the Hilbert space $\mathcal{L}^2(\mathbf{G})$ of square integrable functions on $\mathbf{G}$ where the functions $f$ are represented by means of the left-regular representation $(\lambda(f)\psi)(\beta)= \sum_{t(\alpha)=t(\beta)}f(\alpha)\psi(\alpha^{-1}\circ \beta)$.    

In any case the algebra $C^*(\mathbf{G})$ carries a $C^*$-algebra structure that distinguishes its real (or self-adjoint) elements, i.e., those ones such that $f^* = f$, and that can be identified with physical observables. Note that in such a case, $f(\mathbf{1}_x)=f(x)$ can be interpreted as the actual value of the observable $f$ when the system is ``selected'' at $x\in \Omega$ and $f(\alpha)$ will be associated with the transition amplitude of the observable $f$ between the events $x$ and $y$ connected by the morphism $\alpha:\,x\, \rightarrow \, y$. We will consider our $C^*$-algebra to be unital, i.e. $\mathbf{1} \in C^*(\mathbf{G})$. 

States of the system, on the other hand, will be determined as positive normalized linear functionals $\rho:\,C^*(\mathbf{G})\rightarrow\, \mathbb{C}$, i.e. $\rho(f^* f) \geq 0$ and $\rho(\mathbf{1}) = 1$. This entails the natural statistical interpretation of states as providing the expectation values of observables. Indeed, the number $\rho(f)$ will be interpreted as the expectation value of the observable $f$ (note that if $f$ is real, then $\rho(f)$ is a real number, too):   
\begin{equation}
\left\langle f \right\rangle_{\rho} = \rho(f) = \sum_{\alpha \in \mathbf{G} } f(\alpha) \rho(\delta_{\alpha})\,.
\label{eq:expectation value}
\end{equation}
Denoting by $\phi_{\rho}: \, \mathbf{G}\, \rightarrow \, \mathbb{C}$ the function defined by the state $\rho$ on $\mathbf{G}$ (again, assuming that the groupoid $\mathbf{G}$ is discrete) as:
\be
\phi_{\rho}(\alpha):=\rho(\delta_{\alpha}),
\ee
then $\left\langle f \right\rangle_{\rho} = \sum_{\alpha \in \mathbf{G}} f(\alpha)\phi_{\rho}(\alpha)$ and the expected value of the observable $f$ can be understood as the mean value of the amplitude $f(\alpha)$ over the groupoid $\mathbf{G}$ with respect to the ``distribution'' $\phi_{\rho}$. The function $\phi_{\rho}$ associated to the state $\rho$ is positive definite, that is:
\begin{equation}
\phi_{\rho}\left(\left(\textstyle{\sum_{i=1}^{N}\zeta_{i}\alpha_{i}}\right)^{*}\,\left(\textstyle{\sum_{j=1}^{N}\zeta_{j}\alpha_{j}}\right)\right)=\sum_{i,j=1}^N \overline{\zeta}_i \zeta_j \phi_{\rho}(\alpha_i^{-1}\circ \alpha_j) \geq 0\, , 
\end{equation} 
for all $N\in \mathbb{N}$, $\zeta_i \in \mathbb{C}$, $\alpha_i \in \mathbf{G}$, when the sum is taken over pairs, $\alpha_i, \alpha_j$ with $\alpha_i^{-1}, \alpha_j$ composable. Notice that, if in addition $f$ is a real observable, Eq.\eqref{eq:expectation value} implies $\rho(f)= \overline{\rho(f)}$, then $\phi_{\rho}(\alpha) = \overline{\phi_{\rho}(\alpha^{-1})}$ and $\phi_{\rho}$ itself is an observable. If in addition $\phi_{\rho}(\alpha^{-1})=\phi_{\rho}(\alpha)^{-1}$ we will call this property of the state $\rho$ ``unitarity''. 
The function $\phi_{\rho}$ characterizes the state $\rho$ and allows us to interpret of the state as providing the probability amplitudes of transitions.     

This kinematical background is completed by a dynamical evolution in Heisenberg's form \cite{Ib18b}, \cite{Ib20}, and a theory of transformations that is encoded in a 2-groupoid structure $\boldsymbol{\Gamma}  \rightrightarrows \mathbf{G} \rightrightarrows \Omega$ \cite{Ib18a}, where the morphisms $\varphi \colon \alpha \Rightarrow \beta$ of the groupoid $\boldsymbol{\Gamma}  \rightrightarrows \mathbf{G}$ define the physical transformations that the system can suffer.  We will not continue this discussion here and it will be left for an extended discussion on symmetries and transformations that will be  done elsewhere.


\section{Subsystems}\label{sec:subsystems}


As stated in the introduction, a quantum system is described by a 2-groupoid $\boldsymbol{\Gamma}  \rightrightarrows \mathbf{G} \rightrightarrows \Omega$ consisting of `outcomes', `transitions' and `transformations'.  We will concentrate for the purposes of this study on the first layer of the 2-groupoid structure, that is, on the groupoid $\mathbf{G} \rightrightarrows \Omega$, consisting of outcomes $x \in \Omega$ and transitions $\alpha \colon x \to y$.   

It is not obvious, in principle, what could be the meaning of a `subsystem' of a given system.    We could appeal to the notion of a `part' of the original system, provided that a way of splitting it into smaller parts is available. This would be the obvious interpretation when the system under study is `composed' of various pieces clearly identifiable within the experimental devices used in the description of the system.   We may think, for instance, of a gas.   If the experimental setting consists of a thermometer and a family of other  devicess that allow us to measure the pressure, volume, mass, etc., of it, then it doesn't make sense to think that an individual molecule is a subsystem of the given system (there is no way of determining its structure as a proper system, i.e., its outcomes and transitions for instance).   Actually in classical statistical physics, individual molecules are not considered subsystems of the gas at all.  For instance, subsystems are defined, in a painful and ambiguous way, as in Landau's masterful work \cite[Chap.1, page 2]{La80}: \textit{``A part of the system which is very small compared with the whole system but still macroscopic, may be imagined to be separated from the rest;...  Such relatively small but still macroscopic parts will be called subsystems''.}

If you were expecting that subsystems would be required to keep some relative independence or individuality (like the molecules of the gas), you would find yourself wrong \cite[Chap. 1, page 2]{La80}: \textit{``A subsystem is again a [mechanical] system, but not a closed one; it interacts in various ways with the other parts of the system... these interactions will be very complex and intricate''}.  The lack of independence becomes so  relevant that it is precisely this intricacy in the description of the subsystem that allows to introduce statistical methods to study the whole system \cite[Chap. 1, page 3]{La80}: \textit{``A fundamental feature of this approach is the fact that, because of the extreme complexity of the external interactions with the other parts of the system, will be many times in every possible state''}.  

The previous example from classical statistical physics, shows that the notion of subsystems can be used for many purposes and presented in different forms, not always with the reductionistic idea that we can split the system in smaller pieces (individual and independent in some sense) from which the original system is reconstructed, i.e, that the system is a composition of various subsystems.  Some of these aspects involved in the notion of composition and independence will be discussed at length later on, but only after the notion of subsystem has been properly elucidated because, at is should be clear from the previous comments, it is the key notion that sustains all other elaborations.  

\subsection{Subsystems and subgroupoids}

The development of quantum information has made the discussion on the nature of subsystems even more relevant as `flexible' decomposition of quantum systems into subsystems is crucial to identify degrees of freedom that can be treated as qubits.  For instance, Viola, Knill, and Laflamme \cite{Vi01} and Zanardi, Lidar, and Lloyd \cite{Za04} proposed, as we did before, that the partition of a system into subsystems should depend on which operations are experimentally accessible.   An elaborated conceptual framework encompassing these ideas has been put forward by Chiribella \cite{Chi19} (introducing the notion of an auxiliary agent):  \textit{`The core of our construction is to associate subsystems to sets of operations, rather than observables. To fix ideas, it is helpful to think that the operations can be performed by some agent. Given a set of operations, the construction extracts the degrees of freedom that are acted upon only by those operations, identifying a ``private space'' that only the agent can access. Such a private space then becomes the subsystem, equipped with its own set of states and its own set of operations'.}

Without pretending a complete translation of Chiribella's abstract general framework into the groupoids picture of quantum mechanical systems, it is easy to identify the main ingredients of the notion of subsystem.  Subsystems will be systems, that is, groupoids themselves with their own spaces of outcomes and transitions, that will only contain a certain subset of outcomes and transitions of the original system. This will constitute the `private space' determined by an auxiliary agent in Chiribella's description.   This notion is exactly the mathematical concept of subgroupoid of the given groupoid.  Hence, we can summarize this, admittedly rather long, discussion by stating that subsystems of a given system are subgroupoids.  
In the following paragraphs we will develop some simple consequences of this notion and we will illustrate it with a few relevant examples but, first, we will state the previous definition more formally.

\begin{definition}
Let $\mathbf{G} \rightrightarrows \Omega$ be a groupoid describing a quantum system.  A subgroupoid of the groupoid $\mathbf{G}$ is a subset $\mathbf{H} \subset \mathbf{G}$  such that it becomes a groupoid itself with the composition law inherited from $\mathbf{G}$. Any subgroupoid $\mathbf{H} \subset \mathbf{G}$ will be called a subsystem of the given quantum system.  
\end{definition}

Elaborating on the previous definition (see, for instance, \cite{Ib18,Ib19} for more details on the notion of subgroupoid), if we denote by  $\Lambda   \subset \Omega$, the subset of objects corresponding to units $1_x \in \mathbf{H}$, for any $\alpha \colon x \to y \in \mathbf{H}$, then the units $1_x$, $1_y$ also belong to $\mathbf{H}$ and, consequently $x,y \in \Lambda$.  It is also a consequence of the definition that if $\alpha \colon x \to y \in \mathbf{H}$, then the inverse transition $\alpha^{-1} \colon y \to x \in \mathbf{H}$ too.  Moreover, if $\alpha \colon x \to y$, $\beta \colon y \to z$ are both transitions in the subgroupoid $\mathbf{H}$, then $\beta \circ \alpha \colon x \to z$ must be in $\mathbf{H}$.   Finally, we observe that the associativity law is automatically satisfied as the composition law of the subgroupoid is just the same as that of the groupoid.  The restriction of the source and target maps of the groupoid $\mathbf{G}$ to the subgroupoid $\mathbf{H}$ define, respectively, the source and target maps for the subgroupoid that, accordingly, can be properly denoted as $\mathbf{H} \rightrightarrows \Lambda$.  In what follows we will use both terminologies, subgroupoid or subsystem, interchangeably.

A few examples with help clarifying the notion we are introducing.

\subsubsection{Subsystems of the harmonic oscillator}\label{ex:oscillatorN}

This example describes any quantum system which possesses only a countable set of discrete outcomes.

 Let $\mathbf{A}_\infty$ be the groupoid defining the harmonic oscillator (see \cite{Ib18b} for details).   This groupoid is the groupoid of pairs of the set of natural numbers $\mathbb{N}$. Outcomes are just natural numbers (indicating the energy level of the system) and transitions $\alpha \colon m \to n$, are just pairs of natural numbers $(n,m)$ denoting the physical transition from energy level $m$ to $n$.
If we select a subset $\Lambda \subset \mathbb{N}$ of energy levels, the groupoid of pairs of $\Lambda$ defines a subgroupoid of $\mathbf{A}_\infty$.   For instance, we may choose $\Lambda = \{1,2, \ldots, N \}$, to be the first $N$ levels, hence the corresponding subsystem will be the finite level approximation of the harmonic oscillator obtained when restricting our attention to its first $N$ levels.

We will also notice immediately that the algebra of a subgroupoid is in a natural way a subalgebra of the original groupoid.  In the previous example, for instance, the algebra of the subgroupoid is the algebra $M_N(\mathbb{C})$ of $N\times N$ matrices, which is a subalgebra of the algebra of the groupoid of the harmonic oscillator, i.e.,  the $C^*$-algebra of bounded infinite-dimensional matrices $M_\infty(\mathbb{C})$ (notice that the choice of the subgroupoid fixes the embedding $M_N(\mathbb{C}) \to M_\infty(\mathbb{C})$ which depends on the chosen subgroupoid). 

\subsubsection{Classical systems as subgroupoids} 

Let $\Omega$ denote a set determining the outcomes of a given system, for instance, $\Omega$ can denote a box where a given particle is located and the outcomes $x \in \Omega$ denote the coordinates of the position of the particle (determined by using a detector).   We assume that a topology related to the apparatuses used to detect the particle is introduced.  That is, each detector determines a subset $U_a$ such that the detector is triggered whenever the particle appears in $U_a$.  The family of all such `triggering events' $U_a$ can be used as a sub-basis for a topology on $\Omega$.  We will assume that $\Omega$ is compact with respect to this topology (this will happen, for instance, if the family of subsets $U_a$ provided by the detectors is finite).     It is also natural to consider that the `triggering' sets $U_a$ have different `sizes', say $\mu(U_a)$,  determined by using additional devices in our experimental setting (for instance, an arbitrary tape).  This will be mathematically described by assuming that we are using a (standard) measure $\mu$ on the Borelian $\sigma$-algebra defined by the topology introduced on $\Omega$ (that it, this measure is not associated to a physical state of the system, but to the experimental setting used to describe it).   

Then, if we consider the groupoid of pairs of $\Omega$ as the groupoid describing the system, that is, we are assuming that all transitions $\alpha \colon x \to y$, $x,y$ arbitrary points in $\Omega$, are physically possible, the algebra of the groupoid can be constructed by considering the completion of the algebra of the groupoid  (finite formal linear combinations of its elements) with respect to an appropriate norm.  When using the norm induced from the Hilbert space $L^2(\Omega, \mu)$, the $C^*$-algebra of bounded operators on it will be obtained\footnote{We will consider compactly supported continuous functions $\eta$ on the groupoid, that is $\eta$ will be just a function $\eta (x,y)$ and the corresponding element in the groupoid algebra acts on functions in $L^2(\Omega, \mu)$ as a kernel operator, that is $(\pi_0(\eta) \Psi)(x) = \int_\Omega \eta (x,y) \Psi (y) d\mu (y)$.}.   We will call, as usual, the Hilbert space $\mathcal{H}_\Omega = L^2(\Omega, \mu)$, that supports the fundamental representation $\pi_0$ of the groupoid of pairs $\mathbf{G}(\Omega)$, the fundamental Hilbert space of the system.

The subgroupoid given by the isotropy groups of the groupoid, that in this case happen to be all trivial, is a totally disconnected groupoid. Such subgroupoid, as in the example discussed above, can be considered to provide a classical description of the particle in a box (that is, there are no `quantum transitions' at all). Note that the corresponding $C^*$-algebra of the subgroupoid is given by the space $L^\infty (\Omega, \mu)$ (which is the von Neumann algebra generated by the action of compactly supported functions on $\Omega$ as multiplication operators on $L^2(\Omega, \mu)$).     Actually this Abelian $C^*$-algebra  is the natural setting for commutative probability theory  on $\Omega$, in sharp contrast with the non-commutative probability theory characteristic of quantum systems that would be provided by the $C^*$-algebra of bounded operators on $L^2(\Omega, \mu)$ together with a tracial state defined on it (see \cite{Vo94} for more details on this approach).

Thus, from this perspective, classical systems can be considered as subsystems of quantum ones.    In any case, if $L^\infty(\Omega, \mu)$ is the $C^*$-algebra of the totally disconnected groupoid\footnote{The bold notation $\boldsymbol{\Omega}$ helps to disinguish between the groupoid whose transitions are the units $1_x \colon x \to x$, and the space of objects $\Omega$, thus, the proper notation, rather redundant, for such groupoid will be $\boldsymbol{\Omega} \rightrightarrows \Omega$ where both, the source and target maps being the identity map, coincide.} $\boldsymbol{\Omega}$ describing a classical system ``localized'' in the box $\Omega$, we may consider subsystems of it by selecting a part of the box, for instance we can put a barrier inside the box separating it in two parts, a subset $\Lambda \subset \Omega$ of $\Omega$ and its complementary $\Lambda' = \Omega \backslash \Lambda$.  Then the groupoid of pairs of $\Lambda$ is a subgroupoid of the groupoid of pairs of $\Omega$, and the corresponding totally disconnected subgroupoid $\boldsymbol{\Lambda}$ consisting of the units themselves, is a subsystem of the system $\boldsymbol{\Omega}$.   

The previous considerations are particularly interesting as they are the prototype of an operation performed everytime thermodynamic properties of systems are considered.  It consists of splitting the system introducing partitions, walls, membranes, etc., and are difficult to understand as subsystems of the original system from other perspectives. They are however subsystems, i.e., subgroupoids (quite natural indeed), when considered from the groupoid perspective.


\subsection{Subsystems and their algebras}\label{sec:subsystems_subalgebras}

Given a subsystem $\mathbf{H} \subset \mathbf{G}$ of the physical system described by the groupoid $\mathbf{G} \rightrightarrows \Omega$, the algebra of $\mathbf{H}$ is a subalgebra of the algebra of the total system.    In the finite situation, where the $C^*$-algebra of the groupoid $\mathbf{G}$ coincides with its algebraic counterpart $\mathbb{C}[\mathbf{G}]$, the previous statement is obvious, the canonical identification being $\boldsymbol{a} = \sum_{\beta\in \mathbf{H}} a_\beta \beta$ maps to itself but now considered as an element in $\mathbb{C}[\mathbf{G}]$.  

A slightly more elaborated argument shows that the same is true in the case of discrete countable groupoids (in such case, the $C^*$-algebra of the groupoid being defined as in Sect. \ref{sec:review}) and the result in the general situation can be obtained by realizing that the space of compactly supported functions in $\mathbf{H}$ can be embedded naturally in $C_c(\mathbf{G})$ by extending them trivially to $\mathbf{G}$.  In any case, given a subgroupoid $\mathbf{H}$ of the groupoid $\mathbf{G}$, we will assume that there is a natural identification of the $C^*$-algebra $C^*(\mathbf{H})$ as a $C^*$-subalgebra of the $C^*$-algebra $C^*(\mathbf{G})$ of the groupoid $\mathbf{G}$.   An important observation is that the unit $\mathbf{1}_\mathbf{H}$ will not be mapped, in general, into the unit $\mathbf{1}_\mathbf{G}$.  The two units, $\mathbf{1}_\mathbf{H}$, $\mathbf{1}_\mathbf{G}$, will be identified only if the space of outcomes of the subgroupoid $\mathbf{H}$ coincides with the space of outcomes of $\mathbf{G}$.

It is important to realize that the converse of this statment is not true. There may exist subalgebras of the algebra of a system which are not the algebras of any subgroupoid, that is, that do not correspond to a subsystem of the original system (see, for instance, the example discussed right after this paragraph).   We may insist on calling them `subsystems', but they will only be considered as such from a very detached point of view, because if they are not algebras of subgroupoids, the physical interpretation of their structure in terms of physical outcomes and transitions will be impossible, so they do not retain any relation to the original physical system we are studying.  

Consider for instance the simple example of the qubit whose groupoid description was carefuly done in \cite{Ib18a, Ib18b} using the $\mathbf{A}_2$ groupoid.  The algebra of the groupoid $\mathbf{A}_2$ of pairs of two objects $\{ +,-\}$ is the  $C^*$-algebra $M_2(\mathbb{C})$ of $2\times 2$ matrices.   We may consider the subalgebra of upper triangular matrices with unit diagonal, which is not a $C^{*}$-algebra .   This algebra is not the algebra of a groupoid (it can be proved that the algebra of a finite groupoid is necessarily semisimple which is not the case for this algebra which is solvable \cite{Ib19}).    We may argue that this algebra is not a $C^*$-subalgebra of the original $C^*$-algebra, however we can extend this example by considering a higher dimensional situation (we will examine other examples further on in this work).

From the previous examples we realize that given a groupoid $\mathbf{G} \rightrightarrows \Omega$ there are some canonical subsystems associated to it.   The fundamental subgroupoid  $\mathbf{G}_0 \rightrightarrows \Omega$, is  the subgroupoid formed by the disjoint union of the isotropy groups of $\mathbf{G}$, that is $\mathbf{G}_0 = \bigsqcup_{x \in \Omega} G_x$ with $G_x = \mathbf{G}(x,x) = s^{-1}(x) \cap t^{-1}(x) = \{ \alpha \colon x \to x \}$ being the isotropy group at $x$.    The subsystem $\mathbf{G}_0$ is totally disconnected (there are no transitions between different outcomes), it behaves classically but it retains part of the noncommutativity of the original system in the inner structure provided by the isotropy groups.   The algebra of the fundamental subgroupoid is the algebra defined by the direct sum of the algebras of the isotropy groups, i.e., $C^*(\mathbf{G}_0) = \bigoplus_{x \in \Omega} C^*(G_x) \subset C^*(\mathbf{G})$\footnote{In the continuous situation, the direct sum must be replaced by a direct integral.}.    Elements in the algebra of the fundamental subgroupoid have the form $\boldsymbol{a} = \boldsymbol{a}_{x_1} \oplus \cdots \oplus \boldsymbol{a}_{x_r}$, $\boldsymbol{a}_{x_k}  \in C^*(G_{x_k})$, $k = 1, \ldots, r$, and the multiplication is componentwise.   Note that $\boldsymbol{a}_x  \cdot \boldsymbol{a}_y = \boldsymbol{a}_y  \cdot \boldsymbol{a}_x = 0$, for all $x \neq y$, $x,y \in \Omega$.

A further reduction of the original system is obtained when we consider the subgroupoid of the fundamental groupoid obtained by taking just the units $\mathbf{1}_x \in G_x$.   This subgroupoid, denoted as in the previous examples by $\boldsymbol{\Omega} \rightrightarrows \Omega$, is the classical counterpart of the original system, where all trace of non-commutativity has disappeared.  As it was discussed in the examples above, provided that $\Omega$ is a measure space, the algebra of the subgroupoid $\boldsymbol{\Omega} $ can be defined as $C^*(\boldsymbol{\Omega} ) = L^\infty(\Omega)$.  

In the previous situations, our subsystems share the space of objects $\Omega$ with the original system.  In such case we will say that the subsystems are full.    This, however, doesn't have to be the case.  For instance, any isotropy group $G_x$ is a subgroupoid of the original groupoid $\mathbf{G}$, in this case over the space of objects consisting of just one object, the outcome $\{ x \}$.   This subsystem focus our attention on the inner structure of the outcome $x$.   

A more interesting situation happens when we select a subset $\Lambda \subset \Omega$ of outcomes.  We can define the subgroupoid $\mathbf{G}_\Lambda \rightrightarrows \Lambda$ of the original groupoid by selecting only those transitions $\alpha \colon x \to y$, where both $x,y$ are in $\Lambda$, that is:
$$
\mathbf{G}_\Lambda = \{ \alpha \colon x \to y \mid x,y \in \Lambda \} \, .
$$
Clearly $\mathbf{G}_\Lambda$ is a subgroupoid of $\mathbf{G}$ that will be called the restriction of $\mathbf{G}$ to $\Lambda$.  The physical interpretation of this subsystem, is that we are focusing our attention only on the outcomes in $\Lambda$ as in Example \ref{ex:oscillatorN}.  For instance if $\Lambda = \{ x \}$, then $\mathbf{G}_\Lambda = G_x$ as discussed in the paragraph above.

This is an important operation as it allows to decompose the groupoid in smaller subsystems. Thus consider a partition $\Omega = \cup_{i}\Omega_i$ of the space of outcomes $\Omega$ on the parts $\Omega_i$. Then we may consider the subgroupoids $\mathbf{G}_{\Omega_i}$ defined by the restriction of $\mathbf{G}$ to $\Omega_i$.   Then each subsystem $\mathbf{G}_{\Omega_i}$ defines a subalgebra $C^*(\mathbf{G}_{\Omega_i})$ of the algebra $C^*(\mathbf{G})$ of the groupoid $\mathbf{G}$.  Can the algebra $C^*(\mathbf{G})$ be reconstructed from the subalgebras $C^*(\mathbf{G}_{\Omega_i})$? or, in other words, can the system described by the groupoid $\mathbf{G}$ be reconstructed from the subsystems $\mathbf{G}_i$?  We will answer this question in Sec. \ref{sec: States and subsystems}.


\subsection{States and subsystems}\label{sec: States and subsystems}

 Given a state $\rho$ of the system defined by the groupoid $\mathbf{G}$, there is a natural restriction of such state to any subsystem $\mathbf{H} \subset \mathbf{G}$.  Actually the natural definition of the restriction of the state $\rho$ to $\mathbf{H}$ is given by the linear functional $\rho_\mathbf{H} \colon C^*(\mathbf{H}) \to \mathbb{C}$:
\begin{equation}\label{eq:conditional}
\rho_\mathbf{H}(\boldsymbol{a}) = \rho(\boldsymbol{a}\mid \mathbf{H}) = \frac{\rho(\boldsymbol{a})}{\rho(\mathbf{1}_\mathbf{H})} \, , \qquad \forall \boldsymbol{a} \in C^*(\mathbf{H}) \, .
\end{equation}
Both, the positivity of $\rho_\mathbf{H}$ and its normalization are easily checked. Note the importance of the normalizing factor $\rho(\mathbf{1}_\mathbf{H})$ because, in general, $\mathbf{1}_\mathbf{H}$ is not the unit of $C^*(\mathbf{G})$.  This observation will become particularly relevant in the discussion of the direct product of systems.

The restricted state $\rho_{\mathbf{H}}$ has a natural interpretation as a `conditional probability', that is, Eq. (\ref{eq:conditional}) is the natural extension of the notion of conditional probability to states of algebras of groupoids and their subsystems.  Thus, in this sense, we can say that $\rho(\boldsymbol{a} \mid \mathbf{H})$ is the conditional expectation value of $\boldsymbol{a}$ with respect to $\mathbf{H}$.  

It is not possible, however, to provide a natural extension of a state of a subsystem to the total one.   This is thoroughly consistent with the statistical interpretation of states, as it is not possible to obtain the statistical properties of the total system from that of a subsystem.  Nevertheless, it is possible to extend a state $\sigma \colon C^*(\mathbf{H}) \longrightarrow \mathbb{C}$ to $C^*(\mathbf{G})$, albeit non canonically.   In the following sections some specific choices related to various compositions will be discussed. 


\section{Composition: free and direct}\label{sec:composition}


\subsection{Direct composition}\label{sec:direct}

As commented in the introduction, direct composition corresponds to the direct product of groupoids.   If $\mathbf{G}_a \rightrightarrows \Omega_a$, $a = 1,2$, are two groupoids, its direct product is the groupoid $\mathbf{G}_1 \times \mathbf{G}_2 \rightrightarrows \Omega_1 \times \Omega_2$, whose outcomes are pairs $(x_1, x_2)$, $x_a \in \Omega_a$, and whose transitions are pairs $(\alpha_1, \alpha_2) \colon (x_1, x_2) \to (y_1, y_2)$, $\alpha_a \colon x_a \to y_a \in \mathbf{G}_a$, $a = 1,2$. 

Clearly $\mathbf{G}_1 \times \mathbf{G}_2$ is a groupoid with the obvious source and target maps $s (\alpha_1, \alpha_2) = (x_1, x_2)$, $t(\alpha_1, \alpha_2) = (y_1,y_2)$ and composition law $(\beta_1, \beta_2) \circ (\alpha_1, \alpha_2) = (\beta_1 \circ \alpha_1, \beta_2 \circ \alpha_2)$, provided that $\alpha_a, \beta_a$ are composable.  

This mathematical operation provides the most naive physical implementation of Lieb-Yngvason definition of `compound system'.  The two systems lying side-by-side on the laboratory table without any dynamical interaction between them, so that the original experimental setups that allows as to observe the outcomes and transitions of both systems separately, can be performed simulaneously\footnote{It may be argued that a specific measuring device was being used on the two systems, so it cannot be used to determined the outcomes of both systems together, but the solution is simple, and only requires to increase the budget and use another identical device.} and the outcomes we witness are pairs $(x_1, x_2)$ of outcomes of $\mathbf{G}_1$ and $\mathbf{G}_2$, and the transitions we observe are pairs $(\alpha_1, \alpha_2)$ of transitions $\alpha_a$ of each individual system. 
This seems clear enough and has been considered as the natural notion of composition of systems both in classical and quantum physics however, as it will be discussed in the following paragraphs, we should be cautious when using the direct composition of systems when analyzing specific experiments.

Note first that the groupoid algebra of the direct product of groups is the tensor product of the corresponding groupoid algebras.   It is trivial to check it in the finite case as $\mathbb{C}[\mathbf{G}_1 \times \mathbf{G}_2]$ consists of finite linear combinations of elements in $\mathbf{G}_1 \times \mathbf{G}_2$, i.e., $\boldsymbol{a} = \sum_{\alpha_1, \alpha_2} a_{\alpha_1, \alpha_2} (\alpha_1, \alpha_2)$. Then a basis of $\mathbb{C}[\mathbf{G}_1 \times \mathbf{G}_2]$ is given by pairs $(\alpha_1, \alpha_2)$, $\alpha_a \in \mathbf{G}_a$, which is also a basis of the tensor product $\mathbb{C}[\mathbf{G}_1]   \otimes \mathbb{C}[\mathbf{G}_2]$ and will be denoted then by $\alpha_1 \otimes \alpha_2$.

The unit elements in the algebras $\mathbb{C}[\mathbf{G}_a]$, $a = 1,2$, and $\mathbb{C}[\mathbf{G}_1]   \otimes \mathbb{C}[\mathbf{G}_2]$, are given respectively by $\mathbf{1}_a = \sum_{x_a \in \Omega_a} \mathbf{1}_{x_a}$, $a = 1,2$, and 
$$
\mathbf{1}_{12} = \sum_{(x_1,x_2) \in \Omega_1 \times \Omega_2} \mathbf{1}_{(x_1, x_2)} = \sum_{x_1\in \Omega_1,x_2 \in \Omega_2} \mathbf{1}_{x_1} \otimes \mathbf{1}_{x_2} = \mathbf{1}_1 \otimes \mathbf{1}_2  \, .
$$ 
Then, clearly, the algebras $\mathbb{C}[\mathbf{G}_a]$, can be considered as subalgebras of $\mathbb{C}[\mathbf{G}_1]   \otimes \mathbb{C}[\mathbf{G}_2]$ identifying them with $\mathbb{C}[\mathbf{G}_1] \otimes \mathbf{1}_2$ and $ \mathbf{1}_1 \otimes \mathbb{C}[\mathbf{G}_2]$, respectively.   

We might expect that the subalgebras  $\mathbb{C}[\mathbf{G}_1] \otimes \mathbf{1}_2$ and $ \mathbf{1}_1 \otimes \mathbb{C}[\mathbf{G}_2]$, would be interpreted as the algebras of two subsystems, i.e., with the algebras of two subgroupoids, and that these two subgroupoids could be identified with the original groupoids $\mathbf{G}_a$, $a = 1,2$.  However, the situation is not so simple.   

In fact, the groupoid $\mathbf{G}_1$ is not a subgropoid of $\mathbf{G}_1\times \mathbf{G}_2$, very much as $\mathbb{C}[\mathbf{G}_1]$ is not a subalgebra of $\mathbb{C}[\mathbf{G}_1]   \otimes \mathbb{C}[\mathbf{G}_2]$.   The identification of the algebra $\mathbb{C}[\mathbf{G}_1]$ with the subalgebra $\mathbb{C}[\mathbf{G}_1] \otimes \mathbf{1}_2$ helps to understand what is the subsystem we are considering.   Consider the subgroupoid $\mathbf{G}_1 \times \boldsymbol{\Omega}_2$, i.e., the groupoid of pairs $(\alpha_1, 1_{x_2})$, $\alpha_1 \colon x_1 \to y_1 \in \mathbf{G}_1$, $x_2 \in \Omega_2$. 
Similarly, we may consider the subgroupoid $\boldsymbol{\Omega}_1 \times \mathbf{G}_2 \subset \mathbf{G}_1 \times \mathbf{G}_2$.  The algebra of the groupoid $\boldsymbol{\Omega}_1 \times \mathbf{G}_2$ is given by $\mathbb{C}[\boldsymbol{\Omega}_1 \times \mathbf{G}_2] = \mathbb{C}[\boldsymbol{\Omega}_1] \otimes \mathbb{C}[\mathbf{G}_2] = \mathbb{C}^{n_1} \otimes  \mathbb{C}[\mathbf{G}_2]$, with $n_1 = |\Omega_1|$.    The algebra $\mathbf{1}_1 \otimes  \mathbb{C}[\mathbf{G}_2]$  is a subalgebra of  $\mathbb{C}^{n_1} \otimes  \mathbb{C}[\mathbf{G}_2]$ and it is the subalgebra usually considered as describing a `subsystem' of the total system.   

This is a  particularly interesting instance of the situation discussed in Sect. \ref{sec:subsystems_subalgebras} where a subalgebra of the algebra of a groupoid is not the algebra of a subgroupoid, so it cannot be considered as determining a subsystem of the original system.
This observation, as anticipated before has important consequences in the interpretation of relevant physical experiments (see later on the discussion of the EPRB experiment, Sect. \ref{sec:eprb}).

We will end the discussion on the direct composition of systems by observing that because the original systems defined by the groupoids $\mathbf{G}_a$ are not subsystems of the composite systsem $\mathbf{G}_1 \times \mathbf{G}_2$ until we made a specific choice of the identification of $\mathbf{G}_a$ with a subsystem in $\mathbf{G}_1 \times \mathbf{G}_2$, then it is not possible to just restrict a state $\rho$ of the total system to the individual systems $\mathbf{G}_a$.    This implies that the statistical properties derived from the state $\rho$ do not have a direct interpretation in terms of statistical properties of $\mathbf{G}_a$ separately.    Once a specific choice for the subsystem corresponding to $\mathbf{G}_1$ has been made, for instance $\mathbf{G}_1\times \boldsymbol{\Omega}_2$, then we may consider the restriction of the state $\rho$ to the subalgebra $C^*(\mathbf{G}_1)\otimes \mathbb{C}^{n_2}$, and im particular to the element $\mathbf{1}_2 \in  \mathbb{C}^{n_2}$.   This specific choice is commonly called `tracing out' the second system and, in the particular instance of qudits, the resulting state is denoted as $\mathrm{Tr}_2(\rho)$.   We would like to emphasize again that this is not by any means the only possibility, as the choice of any subset $\Lambda_2 \subset \Omega_2$ will define another identification (as good as the previous one), and the corresponding statistical interpretation will differ.    
We will end this discussion after the remarks in Sect. \ref{sec:universal}.


\subsection{Free composition}\label{free-composition}

After the discussion in the previous section we may wonder if there is another composition operation that will be more natural with respect to the notion of subsystem introduced before, that is, is there a notion of composition such that the original systems become subsystems of the composition and that no additional relations among them appear?  There is, indeed, a well-known mathematical construction that in the case of groups does exactly this.  It is known under the name of free product of groups and we are going to extend it to the category of groupoids.    

Consider the family of groupoids $\mathbf{G}_a \rightrightarrows \Omega_a$, the index $a$ describes a family that does not have to be finite even if, for the purposes of this work and for simplicity, we may consider it to be finite.    It will be convenient to consider an additional set of outcomes $\Omega$ and a family of injective maps $\sigma_a \colon \Omega_a \to \Omega$, that identify the outcomes $x_a$ corresponding to each system $ \mathbf{G}_a$ with an element $\sigma_a(x_a) \in \Omega$.  Thus, we may think of $\Omega$ as providing a common identification for the outcomes corresponding to the various experimental settings associated to the groupoids $\mathbf{G}_a$. 

Then we will construct the family of finite consistent words $w = \alpha_r \alpha_{r-1} \cdots \alpha_2 \alpha_1$, where $\alpha_k \colon x_k \to y_k \in \mathbf{G}_{a_k}$, $k = 1, \ldots, r$, and consistent means that $\sigma_{a_k}(y_k) = \sigma_{a_{k-1}} (x_{k+1})$, $k = 1, \ldots, r-1$ or, in other words, that the sources and targets of consecutive transitions in a given word agree in the common backgraound provided by the space $\Omega$.   Words will also be called `histories' and the individual elements $\alpha_k$ will be called letters or steps accordingly.   

A word or history $w$ will be said to be reduced whenever two consecutive steps that belong to the same groupoid are composed. Note that with this convention the units $1_{x_k}$ will never appear as a letter of a words except when considered as an individual word.   Thus, for instance the reduced word corresponding to the word $\tilde{w} = 1_{y_r} \alpha_r \alpha_r^{-1} \alpha_{r-1}  \cdots 1_{y_2} \alpha_2 1_{x_2} \alpha_1 1_{x_1}$, is the word $w' = \alpha_{r_1} \cdots \alpha_1$.   We will denote the reduced word of a given word by $(w)_\mathrm{red}$, thus, in the previous example $(\tilde{w})_\mathrm{red}= w'$.  It is relevant to emphasize again that consecutive steps of a reduced word always  belong to different groupoids, that is, $w_{\mathrm{red}} = \alpha_r \cdots \alpha_1$, then 
$\mathbf{G}_{a_{k+1}} \neq \mathbf{G}_{a_{k}}$, $k = 1, \ldots, r-1$.   We will call a history consisting of just one step trivial.   

We will call the space of reduced words with steps in the family of groupoid $\mathbf{G}_a \rightrightarrows \Omega_a$, and identification maps $\sigma_a \colon \Omega_a \to \Omega$, the free product of the groupoids $\mathbf{G}_a$ with respect to the family of maps $\{ \sigma_a\}$,  and it will be denoted as $\star_{\sigma_a} \mathbf{G}_a$ (or just $\star_a \mathbf{G}_a$ if there is no risk of confussion). There are natural source and target maps $s,t \colon \star_{a}\mathbf{G}_a \to \bigcup_a \Omega_a$, given respectively by $s(w) = s(\alpha_1) = x_1 \in \Omega_{a_1}$ and $t(w) = t(\alpha_r) = y_r \in \Omega_{y_r}$, provided that $w = \alpha_r \cdots \alpha_1$ is a reduced word.   

Note that  the spaces of objects $\Omega_a$ do not have to coincide. The union of all of them  $\bigcup_a \sigma_a(\Omega_a) \subset \Omega$ will be called the object space of the free product and it will be assumed that $\bigcup_a \sigma_a(\Omega_a) = \Omega$.   The composition law that makes the free product of groupoids into a groupoid itself is the natural concatenation of consistent histories followed by reduction, that is, the histories $w$ and $w'$ can be composed if $t(w) = s(w')$, and then we define:
$$
w' \circ w = (w'w)_\mathrm{red} \, .
$$
It is a simple check to show that the composition operation thus defined is associative and that the histories given by the units themselves are the units of this operation.   Finally we observe that given the reduced history $ w = \alpha_r \alpha_{r-1} \cdots \alpha_2 \alpha_1$ its inverse is the history $w^{-1} = \alpha_1^{-1} \alpha_2^{-1}\cdots \alpha_{r-1}^{-1} \alpha_r^{-1}$ (clearly $w^{-1}\circ w = 1_{x_1}$ and $w \circ w^{-1} = 1_{y_r}$).   

The most striking fact about the free product of a family of groupoids in comparison with the direct product of groupoids defined earlier, is that any groupoid $\mathbf{G}_a$ is a subgroupoid of $\star_a \mathbf{G}_a$, the canonical embedding $i_a \colon  \mathbf{G}_a \to \star_a \mathbf{G}_a$ is the map the sends each transition $\alpha_a \colon x_a \to y_a$, into the trivial history $i_a(\alpha_a)$ consisting just of the step $\alpha_a$ (note that the source and target outcomes of the history $i_a(\alpha_a)$ are given by $\sigma_a(x_a)$ and $\sigma_a(y_a)$ respectively).

A few important observations are in order now.   If the groupoids $\mathbf{G}_a$ are ordinary groups, then the free product of them coincides with the well-known free product of groups.  The free product of groups is used, for instance, to describe the homotopy group of a bouquet of topological spaces.  Consider, for instance, as a simple example the bouquet of circles $V_{k = 1}^n S_k$, consisting of the quotient of $n$ circles $S_i$ where a point $p_k \in S_k$ in each one of them has been identified. Then $\pi_1(V_{k = 1}^n S_k)$ is the free product of the homotopy groups of each one of them, i.e., the free product of $\mathbb{Z}$ $n$-times.  It is not hard to convince oneself that the free product of $\mathbb{Z}$ with itself $n$-times is the free group with $n$ generators $\mathbb{F}_n$.

A relevant situation happens when the spaces of outcomes of each groupoid are disjoint.  In such a case there are no nontrivial histories because there are no transitions in different groupoids that are consistent pairwise, hence they cannot form histories.  Moreover, when we form histories with steps in the same groupoid, when reducing them, the steps have to be composed within the groupoid, and the resulting reduced word is trivial.  Thus we have shown that if $\mathbf{G}_1 \rightrightarrows \Omega_1$, $\mathbf{G}_2 \rightrightarrows \Omega_2$, have disjoint spaces of outcomes, $\Omega_1 \cap \Omega_2 = \emptyset$, then the free product $\mathbf{G}_1 \star \mathbf{G}_2$ is just the disjoint union $\mathbf{G}_1 \bigsqcup \mathbf{G}_2$ of both groupoids.   Hence the disjoint union of groupoids is just an instance of the free product.  


\subsection{Direct vs. free composition: Universal properties}\label{sec:universal}

Apart from the different mathematical definitions, the physical properties of the direct and free product of groupoids as operations with physical systems are significantly different and lead to relevant physical interpretations.   Thus, from the point of view of the experimental settings associated to each operation, we have agreed that the direct product of two physical systems corresponds to the operation of putting them `side-by-side' on the laboratory table in such a way that the outcomes of both of them can be simultaneously recorded.   The transitions of both of them will happen without affecting each other but, from the point of view of subsystems, the original systems do not form subsystems of the direct product.  Note that we cannot make significant statistical observations about the first system unless we clearly specify what the second is doing, or in other words, unless we `freeze' the second system in a specific configuration (for instance, the configuration given by $\mathbf{1}_2 = \sum_{x_2 \in \Omega_2} \mathbf{1}_{x_2}$ with uniform probability distribution).     Clearly, this is a choice made by the observers and there is nothing  natural on it.    

Contrary to this, the free product allows to extend the possible outcomes of the system, creating a large system whose outcomes encompass all possible outcomes as well as new transitions.

In any case, and as a matter of future reference, we must point out that both products, direct and free are characterized by universal properties, that is they are the compositions which are consistent with the notion of subsystems (the free product) and with the notion of  coarse-graining or reduction (the direct product).

We have already seen that there is a canonical family of monomorphisms $i_a \colon \mathbf{G}_a \to \star_a \mathbf{G}_a$, then if $\varphi_a \colon \mathbf{G}_a \to \mathbf{G}$ is another family of monomorphisms in the groupoid $\mathbf{G} \rightrightarrows \Omega$ then there is a monomorphism $\Phi \colon \star_a \mathbf{G}_a \to \mathbf{G}$ such that $\Phi\circ i_a = \varphi_a$ for all $a$.   The monomorphism $\Phi$ is defined as $\Phi (\alpha_r \cdots \alpha_1 ) = \varphi (\alpha_r) \circ \cdots \circ \varphi (\alpha_1) $.  In categorical terms we say that a product satisfying the previous property is a coproduct.   

Similarly, given the direct product of groupoids $\prod_a \mathbf{G}_a$, there is a family of canonical epimorphisms of groupoids $\pi_a \colon \prod_a \mathbf{G}_a \to  \mathbf{G}_a$ given by $\pi (\alpha_1, \ldots, , \alpha_a, \ldots) = \alpha_a$. Then, given another family of epimorphisms $\eta_a \colon \mathbf{G} \to \mathbf{G}_a$, there is a groupoid epimorphism $\Phi \colon G \to \prod_a \mathbf{G}_a$, such that  $\pi_a \circ \Psi = \eta_a$ for all $a$.  

The previous discussion illuminates the deep reason behind the different behaviour of states with respect to direct and free products.   Free product behave naturally with respect to immersions (i.e., with respect to subsystems) while the direct product behaves naturally with respect to projections\footnote{In relation with differential geometry, one may think of submanifolds versus quotient manifolds, or covariant tensors versus contravariant tensors.}. Hence states can be restricted to factors defining a free product but not with respect factors defining a direct product. On the other hand, states can be pulled back from factors to direct products and that leads to the notion of separable states.  

Let $\mathbf{G} = \prod_\iota \mathbf{G}_\iota$ be a direct product of groupoids. Then a state $\rho_\iota$ corresponding to the groupoid $\mathbf{G}_\iota$ defines a state $\rho = \pi_\iota^* \rho_\iota$ on $\mathbf{G}$ as:
$$
\rho (\boldsymbol{a}_1 \otimes \cdots \otimes \boldsymbol{a}_\iota \otimes \cdots \otimes \boldsymbol{a}_n ) = \rho_\iota (\boldsymbol{a}_\iota) \, ,
$$
ands we extend it to an arbitary $\boldsymbol{a}\in C^*(\mathbf{G})$ by linearity.
States of the form $\rho = \pi_1^*\rho_1 \otimes \cdots \otimes \pi_n^*\rho_n$ will be called (homogenous) separable and convex combinations of them will be called separable. They are the direct product analogue to factorizable states.

\subsection{Example: the free product $\mathbf{A}^{(a)}_2 \star_{\sigma} \mathbf{A}^{(b)}_2$}

In this subsection, we will present an example of the free composition of two groupoids of pairs $\mathbf{A}^{(a)}_2$ and  $\mathbf{A}^{(b)}_2$, along the maps $\sigma_a$, $j=a,b$, each of them corresponding to a two-level quantum system. 
The direct product of the same groupoids will be analyzed in section \ref{EPRB_experiment}. 

The space of objects of each groupoid $\mathbf{A}^{(j)}_2$ contains two elements, say $ \Omega_j = \left\lbrace +_j,-_j \right\rbrace$, and there are four morphisms, say, $\mathbf{A}^{(j)}_2 = \left\lbrace \mathbf{1}_{+_j}, \mathbf{1}_{-_j}, \alpha_j, \alpha_j^{-1} \right\rbrace$, where $\alpha_j\, \colon \, +_j \, \rightarrow \, -_j$. 
As previously explained, the free product of two groupoids will depend on two injective maps $\sigma_a, \sigma_b$, from the two spaces of objects, $\Omega_a$ and $\Omega_b$, into the space of objects of the composed groupoids,  $\Omega$, respectively. 
The images of these two maps can have zero, one or two common elements, and, according to these choices, we obtain three different composed groupoids. Let us analyze each of these instances separately.

If $\sigma_a(\Omega_a) \cap \sigma_b(\Omega_b) = \varnothing $, then $\mathbf{A}^{(a)}_2 \star_{\sigma_j} \mathbf{A}^{(b)}_2 =  \mathbf{A}^{(a)}_2 \sqcup \mathbf{A}^{(b)}_2$. 
The space of objects $\Omega$ of the free product is the disjoint union $\Omega_a \sqcup \, \Omega_b = \left\lbrace +_a, -_a, +_b, -_b \right\rbrace$, and analogously for the space of morphisms.
 The groupoid $C^*$-algebra $C( \mathbf{A}^{(a)}_2 \star_{\sigma_j} \mathbf{A}^{(b)}_2 )$ is isomorphic to the direct sum of the groupoid $C^*$-algebras of each groupoid of pairs, i.e., to $C(\mathbf{A}^{(a)}_2)\oplus C(\mathbf{A}^{(b)}_2)$.

\vspace{15pt}

Now, let us consider the situation where $\Omega = \left\lbrace +, 0, - \right\rbrace$ with $\sigma_a(+_a)=+$, $\sigma_a(-_a)= 0$ and $\sigma_b(+_b)=0$, $\sigma_b(-_b)=-$.
Since $\sigma_a(\Omega_a)\cap \sigma_b(\Omega_b) = \left\lbrace 0 \right\rbrace$, we can build the reduced words $\left\lbrace \alpha_a, \alpha_b, \alpha_b\alpha_a \right\rbrace$ and their inverses $\left\lbrace \alpha_a^{-1}, \alpha_b^{-1}, \alpha_b^{-1}\alpha_a^{-1} \right\rbrace$, while the units are $\left\lbrace \mathbf{1}_+, \mathbf{1}_0, \mathbf{1}_- \right\rbrace$. 
Therefore, this composed groupoid is homomorphic to the groupoid of pairs $\mathbf{A}_3$ over three elements. 
Its groupoid $C^*$-algebra $C(\mathbf{A}_3)$ is then isomorphic to the algebra of three-dimensional square matrices with complex entries, i.e., $M_3(\mathbb{C})$, which is the algebra associated with a qutrit. 
This unital algebra contains a commutative unital subalgebra $\mathbf{B} = \mathbb{C}^3$ generated by the elements $\left\lbrace \mathbf{1}_+, \mathbf{1}_0, \mathbf{1}_- \right\rbrace$. 
We will see that $C(\mathbf{A}_3)$ can be thought of as the amalgamated free product of two $C^{*}$-algebras, each of which is isomorphic to $C(\mathbf{A}_{2})\oplus\mathbb{C}$, over the common subalgebra $\mathbf{B}=\mathbb{C}^{3}$.

Let us explain briefly the notion of amalgamated free product over a common subalgebra (for more details see for instance \cite{Vo92,Vo85,Pe99,Po93}).
Given two unital $*$-algebras, say $\mathbf{B}_1$ and $\mathbf{B}_2$, with a common subalgebra $\mathbf{B}$, we can construct the amalgamated free product of the two algebras with amalgamation over the subalgebra $\mathbf{B}$, denoted by $\mathbf{B}_1 \star_{B} \mathbf{B}_2$, as follows. The first step is the definition of the vector space with a basis made up of the words 
$$
\left\lbrace b_1b_2b_3\cdots b_n | n\in \mathbb{N}\right\rbrace,$$
where, if $b_j \in \mathbf{B}_1$, then $b_{j+1}\in\mathbf{B}_2 $ and conversely. 
Then, one has to quotient this vector space by the subspace generated by the relations of the form
\begin{eqnarray}
& b_1\cdots b_{j-1}\left( \lambda_0 b^{(0)}_j + \lambda_1b^{(1)}_j\right)b_{j+1}\cdots b_{n} = \nonumber\\
& = \lambda_0 b_1 \cdots b_{j-1}b^{(0)}_jb_{j+1}\cdots b_n + \lambda_1 b_1 \cdots b_{j-1}b^{(1)}_j b_{j+1}\cdots b_n \label{quotient_1}
\end{eqnarray}
and
\begin{eqnarray}
&  b_j \in \mathbf{B} \, \Rightarrow \, b_1\cdots (b_{j-1}b_j)b_{j+1}\cdots b_n = b_1\cdots b_{j-1}(b_j b_{j+1})\cdots b_n \label{quotient_2} \,.
\end{eqnarray}
The multiplication rule between two words consists in their juxtaposition with a subsequent reduction with respect to the above relations. 

Let us now come back the the groupoid $C^*$-algebra $C( \mathbf{A}^{(a)}_2 \star_{\sigma_j} \mathbf{A}^{(b)}_2 )\cong C(\mathbf{A}_3)$. Let $\mathbf{B}_1$ be the unital $C^*$-subalgebra generated by the elements $\left\lbrace \mathbf{1}_+, \mathbf{1}_0, \mathbf{1}_-, \alpha_a, \alpha_a^{-1} \right\rbrace$ and let $\mathbf{B}_2$ be the unital $C^*$-subalgebra generated by $\left\lbrace \mathbf{1}_+, \mathbf{1}_0, \mathbf{1}_-, \alpha_b, \alpha_b^{-1} \right\rbrace$. 
The algebras $\mathbf{B}_1$ and $\mathbf{B}_2$ possess the common unital subalgebra $\mathbf{B} = \mathbb{C}^3$, generated by $\left\lbrace \mathbf{1}_+, \mathbf{1}_0, \mathbf{1}_- \right\rbrace$. 
Referring to the construction of the amalgamated free product over a common subalgebra recalled above, a direct inspection shows that 
\begin{equation}
C( \mathbf{A}^{(a)}_2 \star_{\sigma_j} \mathbf{A}^{(b)}_2 ) = \mathbf{B}_1 \star_B \mathbf{B}_2.
\end{equation} 
Indeed, the relations in Eqs. \eqref{quotient_1} and \eqref{quotient_2}, enforce that the only non-zero words are those for which $(b_1b)b_2 = b_1(b b_2) \neq 0$. In particular, if $b_1 = \alpha_a$ the only non-vanishing word is $\alpha_a\alpha_b$ because Eq.\eqref{quotient_2} imposes that $b_2$ must have a non vanishing product with $\mathbf{1}_0$. 
Finally, it is interesting to note that $\mathbf{B}_1$ is the groupoid $C^*$-algebra $C(\tilde{\mathbf{A}}^{(a)}_2)$, where $\tilde{\mathbf{A}}^{(a)}_2$ is the extended groupoid $i_a(\mathbf{A}^{(a)}_2) \sqcup \left\lbrace \mathbf{1}_- \right\rbrace $ and $i_a$ is the canonical embedding described in section \ref{free-composition}.
Analogously, $\mathbf{B}_2$ is the groupoid $C^*$-algebra $C(\tilde{\mathbf{A}}^{(b)}_2)$, where $\tilde{\mathbf{A}}^{(b)}_2 = i_b(\mathbf{A}^{(b)}_2)\sqcup \left\lbrace \mathbf{1}_+ \right\rbrace$.
Therefore, we can write 
\begin{equation}
C( \mathbf{A}^{(a)}_2 \star_{\sigma_j} \mathbf{A}^{(b)}_2 ) \simeq C(\tilde{\mathbf{A}}^{(a)}_2) \star_B C(\tilde{\mathbf{A}}^{(b)}_2).
\end{equation}

\vspace{15pt}

As a final example, let us consider the situation where $\Omega = \left\lbrace +,- \right\rbrace$. A possible choice for the maps $\sigma_a$ and $\sigma_b$ is the following one:
\begin{eqnarray}
& \sigma_a(+_a) = +\,, \; \sigma_a(-_a) = - \\
& \sigma_b(+_b) = +\,, \; \sigma_b(-_b) = - ,
\end{eqnarray}
so that  $\sigma_a(\Omega_a)\cap \sigma_b(\Omega_b)= \Omega$. 
In this case, there are infinitely many new words that can be constructed since $\alpha_b^{-1}\alpha_a$ is an element of the isotropy group $\mathbf{G}_+$ of the element $+$, whereas $\alpha_a\alpha_b^{-1}$ is an element of the isotropy group $\mathbf{G}_-$ of the element $-$. 
Indeed, the isotropy groups, which must be isomorphic because the groupoid is connected, are isomorphic to the free group generated by one element (in the case of the isotropy group $\mathbf{G}_+$, the element is $\alpha_b^{-1}\alpha_a$), which coincides with the group $\mathbb{Z}$. 
Let $\mathbf{I}$ be the subgroupoid of $\mathbf{A}^{(a)}_2 \star_{\sigma_j} \mathbf{A}^{(b)}_2$ made up of the units $\left\lbrace \mathbf{1}_+ , \mathbf{1}_- \right\rbrace$, and $\mathbf{G}_0$ be the fundamental subgroupoid, i.e., $\mathbf{G}_0=\underset{x\in \Omega}{\bigsqcup} \mathbf{G}_x $. 
Associated with the groupoid $\mathbf{A}_2^{(a)}\star_{\sigma_j} \mathbf{A}_2^{(b)} $, we can consider the fundamental short exact sequence of groupoids \cite{Ib19} :
\begin{equation}
\mathbf{I} \, \rightarrow \, \mathbf{G}_0 \, \rightarrow \, \mathbf{A}_2^{(a)} \star_{\sigma_j} \mathbf{A}_2^{(b)}\, \rightarrow \, \mathbf{A}_2 \, \rightarrow \, \mathbf{1}_{\ast},
\end{equation}       
where $\mathbf{A}_2 $ is the groupoid of pairs over the space of objects $\Omega_{\mathbf{A}_2\star \mathbf{A}_2} $, and $\mathbf{1}_{\ast} $ is the trivial groupoid with only one object and one morphism, say $\mathbf{1}_{\ast} = \left\lbrace e \right\rbrace $. 
This sequence splits because either $i_a$ or $i_b$, the canonical embeddings defined in section \ref{free-composition}, provides a homomorphism of the groupoid of pairs $\mathbf{A}_2$ into the groupoid $ \mathbf{A}_2^{(a)}\star_{\sigma_j} \mathbf{A}_2^{(b)} $. 
Consequently, according to \cite[Theorem 6.2]{Ib19}, one has that: 
\begin{equation}
\mathbf{A}_2^{(a)}\star_{\sigma_j} \mathbf{A}_2^{(b)} \simeq \mathbb{Z}\times \mathbf{A}_2 .
\end{equation}
We can now compute the groupoid $C^*$-algebra associated with this groupoid. 
On the one hand, since $\mathbf{A}_2^{(a)}\star_{\sigma_j} \mathbf{A}_2^{(b)} \simeq \mathbb{Z}\times \mathbf{A}_2$, the $C^{*}$-algebra $C(\mathbf{A}^{(a)}_2 \star_{\sigma_j} \mathbf{A}^{(b)}_2)$ is isomorphic to the tensor product $L(\mathbb{Z}) \otimes C(\mathbf{A}_2) $, where $L(\mathbb{Z}) $ denotes the group-algebra of the free group with one generator. 
On the other hand, following the same reasoning as in the previous example (where $\Omega = \left\lbrace +, 0, - \right\rbrace $), one can see that the groupoid $C^*$-algebra $C(\mathbf{A}^{(a)}_2 \star_{\sigma_j} \mathbf{A}^{(b)}_2)$ is isomorphic to the amalgamated free product $C(\mathbf{A}_2^{(a)})\star_B C(\mathbf{A}_2^{(b)}) $, where $\mathbf{B}$ is  the unital algebra generated by the elements $\left\lbrace \mathbf{1}_+ , \mathbf{1}_- \right\rbrace$. 

\vspace{10pt}

From this last example, we see that, starting from two systems described by finite groupoids and finite-dimensional $C^{*}$-algebras, the composition associated with the free product  may lead to a system characterized by a countable groupoid with an associated infinite-dimensional $C^{*}$-algebra.
The physical interpretation of this construction is thus more involved, and will be left for a future work.
Furthermore, regarding the fact that the $C^{*}$-algebra of the free product of groupoids is the amalgamated free product of $C^{*}$-algebras, the results contained in this subsection are referring only to specific examples, and a more general discussion will be presented in a forthcoming work.   

\section{Independence}

In the previous sections we have discussed the notion of subsystems and their composition under the light of the theory of groupoids, finding that the notion of subsystem corresponds to the concept of subgroupoids and that composition of systems can be performed in various ways, each one with its own physical interpretation: direct product of groupoids or, in  more generality, extensions of groupoids and free products of groupoids (that include disjoint unions).

One of the main contributions of the groupoids picture of quantum mechanics so far is that the statistical interpretation provided by quantum measures is equivalent to the determination of a state of the system, the quantum measure is obtained directly in terms of the decoherence functional defined on the groupoid by means of the characteristic function (or amplitude) of the state (see \cite{Ib18c} for a detailed account of these notions). 

Thus, the statistical interpretation of a quantum system is provided by the non-commutative probability space $(C^*(\mathbf{G}), \rho)$, where $C^*(\mathbf{G})$ is the $C^*$-algebra associated to the groupoid $\mathbf{G}$ and $\rho$ is a state.  This is particularly relevant because it allows to reformulate quantum mechanical notions directly in terms of the corresponding  problems in non-commutative probability theory.     

In this vein, a particularly relevant and subtle physical question is the notion of `independence' of subsystems. As A.N. Kolmogorov wrote: ``\textit{...one of the most important problems in the philosophy of the natural sciences is -- in addition to the well-known one regarding the essence of the concept of probability itself -- to make precise the premises which would make it possible to regard any given real events as independent.}'' \cite[Ch.I, \S 5, p.9]{Ko50}.   More precisely, we will consider a physical system described by the groupoid $\mathbf{G} \rightrightarrows \Omega$ and a family of subsystems $\mathbf{H}_a \rightrightarrows \Lambda_a$ of it, when should we interpret them as being independent?


\subsection{Landau's statistical notion of independence}

In order to provide an answer to this question, let us recall some standard notions of independence.  From a physical perspective, we may refer again to Landau's  insight \cite[Chap. 1, p. 7]{La80}:

 ``\textit{By statistical independence we mean that the state of one subsystem does not affect the probabilities of various states of the other subsystems}.''

The first observation regarding Landau's notion of independence is its statistical nature, that is, it is not an algebraic property of the subsystems but a statistical property associated to specific states of the systems.
 In a classical setting, it follows immediately from Landau's principle that the statistical distribution $\varrho_{12}$ of a composite system is the product $\varrho_1 \varrho_2$ of the statistical distributions of the separate subsystems and the expectation value of the product of two physical quantities $f_a$, $a = 1,2$, of the corresponding subsystems is the product of the expected values of each one separately, that is:
\begin{equation}\label{eq:statistical_independence}
\langle f_1 f_2 \rangle_{\varrho_{12}} = \langle f_1 \rangle_{\varrho_1} \langle f_2 \rangle_{\varrho_2} \,. 
\end{equation}

More formally, consider a classical system described by the Abelian von Neumann algebra $L^\infty (\Omega, \mu)$, of bounded functions on the probability space $(\Omega, \mu)$ (tipically, $\Omega$ is a phase space and $\mu$ is the corresponding Liouville measure).    States of the system are given by probability distributions $\varrho$ on $\Omega$.  The standard statistical interpretation of the state $\varrho$ is  that the configuration of the system is described by a random variable $\xi$ taking values on $\Omega$, and $\rho$ is the distribution function associated to the random variable $\xi$.    

Then, if the system is the direct product of two subsystems, i.e., the classical groupoid $\boldsymbol{\Omega} = \boldsymbol{\Omega}_1 \times \boldsymbol{\Omega}_2$ is the direct product of the classical groupoids $\boldsymbol{\Omega}_a$, $a = 1,2$ (recall the discussion in Sect. \ref{sec:direct}), the state of the two subsystems are determined by probability distributions $\varrho_a$, $a=1,2$, respectively.  Notice that $L^\infty (\Omega, \mu) =  L^\infty (\Omega_1, \mu_1) \otimes L^\infty (\Omega_2, \mu_2) \cong L^\infty (\Omega_1\times \Omega_2,  \mu_1 \times \mu_2)$ and that the algebras of the two subsystems $L^\infty (\Omega_a, \mu_a)$, $a = 1,2$ are commuting subalgebras of the total algebra $L^\infty (\Omega, \mu)$.  

If the random variable $\xi$ is the direct product of two random variables  $\xi_a$, $a = 1,2$, on $\Omega_a$, $a = 1,2$, respectively, $\xi = \xi_1 \times \xi_2$, i.e., the configuration on each phase space $\Omega_a$, is independent of the position in the other, then, clearly the distribution $\varrho_{12}$ associated to $\xi$ is just the product of the distributions associated to $\xi_a$:
$$
\varrho_{12} = \varrho_1 \varrho_2 \, ,
$$
and we get the statistical independence of expectation values on physical quantities defined on each subsystem given by Eq. (\ref{eq:statistical_independence}).

We will provide an answer to the question about the meaning of independence of physical subsystems by adopting Landau's statistical indepedence principle adapting it to the groupoid description of quantum mechanical systems.   The notion of subsystem has been elucidated  already along this work, however the intuitively apparently clear notion that the `state of the subsystem does not affect the probabilities of various states of the other subsystems' needs to be clarified.


\subsection{Independence and non-commutative probability spaces}

 Abstracting the notions involved in Landau's principle to make them suitable for non-commutative situations (not just in the study of quantum systems), has led to various choices and interpretations.  
 
 For instance (see, for instance, \cite{Vo92,Vo94}), let us consider that $(\mathcal{A},\rho)$ is a non-commutative $C^*$-probability space, i.e., $\mathcal{A}$ is a unital $C^*$-algebra and $\rho$ is a state on it, then, given a family of unital subalgebras $\mathcal{A}_\iota$, the usual notion of independence (modeled on the notion of tensor products) is that they commute among themselves and the expected value of products of elements in the algebra are the products of the corresponding expected values (which mimicks the statistical independence principle expressed by Eq. (\ref{eq:statistical_independence})), that is:
 
 \begin{definition}\label{def:independence_usual} Let $\mathcal{A}_\iota$ be a family of subalgebras of the noncommutative probability space  $(\mathcal{A}, \rho)$. We will say that they are independent in the usual (or standard) way, if they commute among themselves and:
 $$
 \rho (a_1 \cdots a_n) = \rho (a_1) \cdots \rho (a_n) \, , \qquad  \forall a_k \in \mathcal{A}_{\iota_k}\, , \quad 1 \leq k \leq n \, ,
 $$
 and  $k \neq l$ implies that $\iota_k \neq \iota_l$.   
 \end{definition}
 
 Clearly, this is the situation we would be considering if $\mathcal{A} = \otimes_\iota \mathcal{A}_\iota$, and the state $\rho$ would have the form, $\rho = \otimes_\iota \rho_\iota$, then $\rho (a_1 \cdots a_n) = \rho_{\iota_1}(a_1)\cdots \rho_{\iota_n}(a_n)$. Notice that the subalgebras $\mathcal{A}_\iota$ (identified as subalgebras of $\mathcal{A}$ with the subalgebras $\mathbf{1}_1 \otimes \cdots \otimes \mathcal{A}_\iota \otimes \cdots \otimes \mathbf{1}_n$ in $\mathcal{A}$) commute among themselves.   

The standard interpretation is that elements in $\mathcal{A}$ describe (non-commutative) random variables, then given two non-commutative random variables (or two sets of random variables), their statistical independence corresponds to their commutativity and that the expected values of their products are just the products of their expected values.   

This physical interpretation is drawn from the quantum mechanical interpretation of the $C^*$-algebra of bounded operators on a Hilbert space and a state on it as a non-commutative $C^*$-probability space, the self-adjoint operators $A$ being considered as non-commutative random variables with the corresponding probabilty distributions defined by the measure $\nu_{A, \rho} (d\lambda) = \mathrm{Tr\, }( \boldsymbol{\rho}\, E_A(d\lambda))$, where $E_A(d\lambda)$ is the spectral measure associated to $A$. Then, the statistical independence of the measures $\nu_{A,\rho}$, $\nu_{B,\rho}$ imply that the operators $A$ and $B$ commute.  Note that in such case $\rho (AB) = \rho(A) \rho (B)$.   Again a tensor product interpretation in terms of subsystems can be made (more easily in fact).

The previous intepretation of independence is not by any means the only possible one.  Again, keeping the conceptual framework provided by non-commutative $C^*$-probability spaces, D. Voiculescu introduced a different notion of independence that has proved to be invaluable in addressing many problems in operator algebras, group theory, etc. \cite{Vo92, Vo94} and that is perfectly suited to the non-commutative setting.   

\begin{definition}\label{def:independence_free}
We will say that a family of unital subalgebras  $\mathcal{A}_\iota$ of the non-commutative probability space $(\mathcal{A}, \rho)$ is free independent (or just free) if $\rho (a_1 \cdots a_n) = 0$ whenever $\rho (a_k) = 0$, $1 \leq k \leq n$, and $a_k \in \mathcal{A}_{\iota_k}$ where consecutive indices $\iota_k \neq \iota_{k+1}$ are distinct.  
\end{definition}

A main feature of free independence is that, like in the case of usual independence, is that the state $\rho$ is completely determined by the restrictions $\rho_\iota = \rho|_{\mathcal{A}_\iota}$ (see later, Prop. \ref{prop:char_ind}).


\subsection{Generalized Independence}

The previous definitions, albeit well established, do not reflect properly Landau's notion of independence of subsystems because a distinguished state $\rho$ is selected which is not specifically related to the given subsystems.    A natural way to proceed was suggested in \cite{Ci19} where the notion of composition and independence was reviewed using minimal assumptions on the structures used in the description of the systems.

In this direction, the first relevant observation is that, in general, we are interested in families of states, that is, given a family of subsystems $\mathbf{G}_i \subset \mathbf{G}$, we would like to consider all states corresponding to each one of the subsystems $\mathbf{G}_i$, that is, following Landau, we will say that the subsystems are statistically independent if the states $\rho_i$ of the subsystem $\mathbf{G}_i$ do not affect the statistical information associated to the states $\rho_j$ of the other subsystems $\mathbf{G}_j$, $i \neq j$.  

Now the second fundamental observation is that, as it was pointed out already in Sect. \ref{sec:subsystems}, if $\rho$ is a state of the system $\mathbf{G}$, then $\rho$ induces a state $\rho_1$ of the subsystem $\mathbf{G}_1$.  The state $\rho_1$ is just the state obtained by restricting $\rho$ to the subalgebra $C^*(\mathbf{G}_1) \subset C^*(\mathbf{G})$ or, alternatively, obtained by restricting the characteristic function (or amplitude) $\varphi$ of the state $\rho$ to the subgroupoid $\mathbf{G}_1 \subset \mathbf{G}$.   However, in general, there is no natural way to extend a state $\rho_1 \colon C^*(\mathbf{G}_1) \to \mathbb{C}$ to the total system.  

In the particular instance that the system is a direct product of groupoids, hence the algebra is a tensor product of unital algebras, we can do that easily, but it is not possible to restrict them to the subsystems unless further assumptions are made.   One  way to overcome this difficulty is by assuming that the states $\mathcal{S}$ of interest to us, that is, those states defined on the algebra spanned by the union of the algebras $\mathcal{A}_i = C^*(\mathbf{G}_i)$ corresponding to the given family of subsystems, are uniquely determined by the states of the separate subsystems.    This can be done, for instance, as suggested in \cite{Ci19}, assuming that there is an injective map $I \colon \mathcal{S}_1 \times \cdots \times \mathcal{S}_n \to \mathcal{S}$.    This, however, could be too restrictive for some purposes and, as it will be shown below, it is preferable to consider a more general approach.  

As it turns out, Voiculescu's notion of free independence is best suited to the specifics of the non-commutative situation once we extend the definition to consider a family of states of interest $\mathcal{S}$.  

\begin{definition}\label{def:independence_gen}
Let $\mathcal{A}_\iota$ a family of unital subalgebras of the unital algebra $\mathcal{A}$ and let $\mathcal{S}$ denote a subset of the space of states of the algebra generated by the union of the algebras $\mathcal{A}_\iota$, called the `states of interest'.   We will say that the subalgebras $\mathcal{A}_\iota$ are free independent with respect to the family $\mathcal{S}$ (or `generalized independent', or just indpendent if there is no risk of confusion) if for any $\rho\in \mathcal{S}$, such that $\rho(a_k) = 0$, $k = 1, \ldots, r$, $a_k \in \mathcal{A}_{\iota_k}$,  then $\rho (a_1\cdots a_r) = 0$, where two consecutive elements, $a_{\iota_k}$, $a_{\iota_{k+1}}$, belong to different subalgebras, $\mathcal{A}_{\iota_k} \neq \mathcal{A}_{\iota_{k+1}}$.   

We will call elements $a = a_1 \ldots a_r$ in the algebra generated by the subalgebras $\mathcal{A}_\iota$ and satisfying the conditions of this definition, i.e., that consecutive elements belong to different subalgebras, reduced elements, and will denote them by $a_\mathrm{red}$ whenever is necessary to emphasize it.
\end{definition}

A few comments are needed at this point.   First, the previous definition can be made more specific by considering for instance $C^*$-algebras or von Neumann algebras.  We will assume in what follows that we restrict our attention to the class relevant in each situation. 

This definition is clearly a natural extension of Voiculescu's notion of free independence, then we could keep this terminology, however we will prefer to call it just `independence' (or generalized independence) because, at is will be shown below, it captures well enough the notion of statistical independence of subsystems we are trying to understand.   

If the algebras $\mathcal{A}_\iota$ are independent in the usual sense, Def. \ref{def:independence_usual}, i.,e, they commute among themselves and the states of interest satisfy $\rho (a_1 \ldots a_r) = \rho(a_1) \ldots \rho (a_r)$, they do not have to be independent in the generalized sense established in Def. \ref{def:independence_gen}.  The reason for that is twofold. On one side, the class of states of interest considered in the definition could be too large.   But even if we restrict our attention to a smaller class of states there may be elements, for instance, of the form $a_1 a_2 a_3 a_4$ where $a_1, a_3 \in \mathcal{A}_1$ and $a_2, a_4 \in \mathcal{A}_2$ for which the usual property of independence is not satisfied while the notion of generalized independence is. For instance, consider two subalgebras $\mathcal{A}_1$ and $\mathcal{A}_2$ such that they are independent in the generalized sense with respect to the state $\rho$. Then  if $\rho(a_i ) = 0$, $i = 1,2,3,4$, $a_1, a_3 \in \mathcal{A}_1$ and $a_2, a_4 \in \mathcal{A}_2$, then $\rho (a_1a_2a_3 a_4) = 0$.  However, if the algebras $\mathcal{A}_1$, $\mathcal{A}_2$ are independent in the usual sense, so they commute, we can rearrange the product $a_1 a_2 a_3 a_4$ as $a_1 a_2 a_3 a_4 = (a_1a_3) (a_2a_4)$ (the elements of $\mathcal{A}_1$ and $\mathcal{A}_2$ commute among themselves) and then $\rho (a_1 a_2 a_3 a_4) =\rho (a_1 a_3) \rho(a_2 a_4)$ because $a_1a_3 \in \mathcal{A}_1$ and $a_2 a_4 \in \mathcal{A}_2$, but it is not true, in general, that $\rho (a_1 a_3) = \rho (a_1) \rho (a_3)$, hence, $\rho (a_1) =\rho (a_3) = 0$, would not imply that $\rho (a_1 a_3) = 0$.  

There is however a partial inverse of the previous statements that hold for generalized independent subalgebras, summarized in the following proposition, that justifies the choice of the definition.

\begin{proposition}\label{prop:char_ind}
 Let $\mathcal{A}_\iota$ be a family of generalized independent unital subalgebras of the unital algebra $\mathcal{A}$. Then the expected values of states of interest $\rho \in \mathcal{S}$ in reduced elements, i.e., $\rho (a_\mathrm{red})$ with $a_\mathrm{red} = a_1 \ldots a_r$ a reduced element, are determined uniquely by the restrictions $\rho_\iota$ of the state $\rho$ to the subalgebras $\mathcal{A}_\iota$.  Moreover, if $a_\mathrm{red} = a_1 \cdots a_r$, is a reduced element such that the factors $a_k\in \mathcal{A}_{a_k}$ are all in \underline{different} subalgebras, then $\rho (a_1 \ldots a_r) = \rho (a_1)\ldots \rho(a_r)$.
\end{proposition}

\begin{proof}
To prove that $\rho (a_1 \ldots a_n)$ depends on the restrictions $\rho_a$ we proceed by induction on $n$.  If $n =1$, the statements of the proposition are obviously true.   

Now consider $n+1$, $\rho$ a state of interest and $a_\mathrm{red} = a_1 \ldots a_n a_{n+1}$ a reduced element with $n+1$ factors.  
Notice that $\rho (a_k - \rho (a_k) \mathbf{1}) = 0$ for all $k$
Then, because of the definition of independence, we get that:
$$
\rho ((a_1 - \rho(a_1) \mathbf{1})(a_2 - \rho(a_2) \mathbf{1}) \cdots (a_n - \rho(a_n) \mathbf{1}) (a_{n+1} - \rho(a_{n+1}) \mathbf{1})) = 0 \, .
$$
Expanding the products and because of the linearity of $\rho$ we get:
\begin{eqnarray}\label{eq:red}
&& \rho (a_1 \ldots a_n a_{n+1}) = \\ && \sum_{m =1}^n \left((-1)^{n-m} \sum_{1\leq i_1 < i_2 < \cdots < i_m\leq n+1} \rho(a_{j_1}) \ldots \rho (a_{j_{n+1-m}}) \, \rho( a_{i_1} \cdots a_{i_m} ) \right) \, , \nonumber
\end{eqnarray}
where $\{ j_1, \ldots, j_{n+1-m} \}$ is the supplementary set of $\{ i_1, \ldots, i_m \}$ in the set of indices $\{ 1, \ldots, n, n+ 1\}$.
Hence, all terms in the r.h.s. of Eq. (\ref{eq:red}) contain terms of the form  $\rho( a_{i_1} \cdots a_{i_m} )$ with $m < n+1$.   Then, for each term $a_{i_1} \cdots a_{i_m}$, we form the corresponding reduced element $(a_{i_1} \cdots a_{i_m})_\mathrm{red} = a'_{l_1} \cdots a'_{l_m'}$, $m' \leq m$ (that is, if two consecutive elements $a_{i_k}, a_{i_{k+1}}$ are in the same subalgebra, we substitute them by their product $a'_k = a_{i_k}a_{i_{k+1}}$) and then, by the induction hypothesis, $\rho( (a'_{i_1} \cdots a_{i_m})_\mathrm{red} )$ is determined by the restrictions $\rho_a$ of $\rho$ to  the subalgebras $\mathcal{A}_\iota$ and the first assertion is proved.

To prove the second assertion, we observe that if all elements $a_k$ appearing in the reduced element $a_\mathrm{red} = a_1 \cdots a_{n+1}$ belong to different subalgebras, then the terms $a_{i_1} \cdots a_{i_m}$ appearing in the r.h.s. of Eq. (\ref{eq:red}) are automatically reduced too.   Then we can apply to them the same argument an expand the expression in the l.h.s. of the following equation as we did before:
$$
\rho ((a_{i_1} - \rho(a_{i_1}) \mathbf{1})(a_{i_2} - \rho(a_{i_2}) \mathbf{1}) \cdots (a_{i_m} - \rho(a_{i_m})\mathbf{1})) = 0 \, .
$$
Repeating, this process as needed we find that, eventually, all terms in the r.h.s. of Eq. (\ref{eq:red}) have the form: $\rho (a_1) \rho(a_2) \ldots \rho (a_{n+1})$.  Then, we conclude that $\rho (a_\mathrm{red}) = C \rho (a_1) \rho(a_2) \ldots \rho (a_{n+1})$ for some integer coefficient $C$.  But applying the previous formula to the reduced element $\mathbf{1}$ we get $C = 1$ and the proof is completed.
\end{proof}

We have thus identified a general condition on families of subalgebras and states, called general independence, which is genuinely non-commutative and that guarantees Landau's statistical independence on families of elements $a_k$ provided that they belong to different subalgebras.  Note that, in the simpler situation of two subalgebras, for any reduced element we get (\ref{eq:statistical_independence}).

The previous consideration are extremely general, no assumptions on the algebras $\mathcal{A}_\iota$ were made, so we can face the question of when and how to construct independent subalgebras in the groupoids setting.  These will be the arguments of the discussion to follow.


\subsection{Independence and composition in the groupoids setting}

As it turns out, if we focus the general discussion on independence in the groupoids setting, we immediately get some interesting results.

Consider a family of subsystems (subgroupoids) $\mathbf{G}_a \rightrightarrows \Omega_a$ of the groupoid $\mathbf{G} \rightrightarrows \Omega$, then we will say that the subsystems $\mathbf{G}_a $ are independent if the corresponding algebras $C^*(\mathbf{G}_a )$ are generalized independent as subalgebras of $C^*(\mathbf{G})$ with respect to a given family of states of interest $\mathcal{S}$.   Notice that because of Prop. \ref{prop:char_ind}, if the subsystems $\mathbf{G}_a \subset \mathbf{G}$ are independent, then the expectation values of the product of observables belonging to different subsystems in any state of interest is the product of the expectation values of the restriction of the state to the corresponding factors.   Moreover the expectation value of any observable, depends solely on the restrictions of the state to the factors (even through, in principle, complicated non-linear expression).  

It may come as a surprise that if the system $\mathbf{G}$ is the free product of the groupoids $\mathbf{G}_a$, then the subsystems $\mathbf{G}_a$ are, in general, not independent with respect to a large class of states of interest, those called factorizable and that, in \cite{Ib18c}, were identified as embodying Feynman's reproducibility and statistical significance \cite{Fe05}.  Let us recall that a a state $\rho$ on the algebra $C^*(\mathbf{G})$ of the groupoid $\mathbf{G}$ with amplitude $\varphi \colon \mathbf{G}$, is called factorizable if $\varphi (\alpha \circ \beta ) = \varphi (\alpha) \varphi (\beta)$, whenever $\alpha$ and $\beta$ are composable.   This is even more striking as it can be easily shown that such states are determined solely by the restriction of their characteristic functions to a generating quiver of the groupoid.  However it is easy to provide simple examples exhibiting the lack of independence of such situation.  Consider for instance a groupoid $\mathbf{G}\colon \Omega$ which is the free product of two groupoids: $\mathbf{G}_1 \rightrightarrows \Omega_1$, $\Omega_1 = \{ -1, 0\}$, $\mathbf{G}_1 = \{ 1_-, 1_0, \alpha_1 , \alpha_1^{-1} \}$, and $\mathbf{G}_2 \rightrightarrows \Omega_2$, $\Omega_2 = \{ 0, 1\}$, $\mathbf{G}_2 = \{ 1_0, 1_+, \alpha_2 , \alpha_2^{-1} \}$. Then consider the factorizable state defined by $\rho_(\alpha_1) = \rho (\alpha_2) = e^{is}$.  Finally consider the elements $\boldsymbol{a}_1 = - 2 1_- + \alpha_1 + \alpha_1^{-1} \in C^*(\mathbf{G}_1)$, and $\boldsymbol{a}_2 = - 2 1_+ + \alpha_2 + \alpha_2^{-1} \in C^*(\mathbf{G}_2)$.   Clearly, $\rho(\boldsymbol{a}_1 ) = \rho(\boldsymbol{a}_2 ) =0$, bbut $\rho (\rho(\boldsymbol{a}_1 )\cdot \rho(\boldsymbol{a}_2 ) = 1$.   

The interpretation of this fact, is that even if the free product is the `freest' groupoid you may construct out of a family of them (in the sense that no additional algebraic relations among the elements of them are introduced), the composability condition is an `intrinsic' condition that relates the different subsystems.   So, even for the states which are more prone to be independent, those we have termed factorizable, the free product is not (statistically) independent.

If we replace the free composition by the direct composition though, we get instead that the family of states of interest are those called separable, that is, those states which are direct products of states naturally induced from the states of the factors:

\begin{theorem}
Let $\mathbf{G}_a \rightrightarrows \Omega_a$  be a family of groupoids defining quantum physical systems.   Consider the system $\mathbf{G} = \prod_a \mathbf{G}_a$ defined by the direct product of them. The algebra of $\mathbf{G}$ is the tensor product of the algebras $C^*(\mathbf{G}_a)$.  Then the subsystems defined by the subgroupoids $\widetilde{\mathbf{G}}_a = \mathbf{G}_a \times \Omega'_a$, $\Omega_a' = \bigcup_{b \neq a} \Omega_b$,  are statistically independent with respect to the family of separable states.
\end{theorem}


\section{The groupoids interpretation of the EPRB experiment}\label{EPRB_experiment}

We will finish this work by analyzing the Einsteiin-Podolski-Rosen-Bohm (EPRB) experiment using the notions developed so far.   

\subsection{The EPRB experiment}\label{sec:eprb}

The EPRB experiment, originally proposed by Einstein-Podolski-Rosen in \cite{Ei35} and subsequently simplified by Bohm \cite{Bo51}, was aimed to show that the standard description of Quantum Mechanics was (logically) incomplete, so that a hidden variables interpretation (for instance) was necessary to restore its logical consistency.  

The  core of the argument (see for instance \cite[Chap. 16.1]{Pa09} for an expository description of the subject) consists of exhibiting  `elements of physical reality' which are present in a given physical system but that contradict its quantum mechanical interpretation, where by `elements of physical reality' it is understood that \cite{Ei35}: `\textit{If, without in any way disturbing a system, we can predict with probability equal to unity the value of a physical quantity, then, independently from our measurement procedure, there exists an element of the physical reality corresponding to this physical quantity}', 

The  system could consist of two particles of spin 1/2 that are in a state in which the total spin is zero as proposed by D. Bohm \cite{Bo51} in an important simplification of the original setting that was not only clarifying the original argument, but because it is also experimentally feasible.  

Thus the system will consist of two particles of spin 1/2 (see Fig. \ref{fig:eprb}) that are in a singlet state (of total spin 0). The particles can be produced by a single particle decay and, after a certain amount of time $T$ the two resulting particles, labelled 1 and 2, are spatially separated, say by a distance $\Delta$, and they do not interact any longer.  If the system is not disturbed, the conservation of angular momentum guarantees that they remain in a singlet state  which can be written as:
$$
| \Psi_\mathrm{EPRB} \rangle = \frac{1}{\sqrt{2}} \left( |+_1 \rangle | -_2 \rangle - | -_1 \rangle | +_2 \rangle  \right) \, ,
$$
where $| \pm_a \rangle$ represent the eigenstates of the projection of the spin operators $\mathbf{s}_a$, $a = 1,2$, along the $z$-axis.
Now, if an observation of the spin component along the $z$-axis of particle 1 gives the outcome +1/2, that of particle 2 along the same direction, performed after a time $\tau < \Delta/c$, will necessarily give the outcome $-1/2$ and viceversa.  This result will remain true for the three components of spin (as the singlet state is rotationally invariant), hence we should attribute elements of physical reality to all components of spin that, because they do not commute, cannot be considered as such.   

In a remarkable series of papers \cite{Sc35} Schr\"odinger criticised the EPR experiment by observing that the classical separability principle, i.e., that the only form of interdependence between systems is the dynamical causal form of interaction, that is, caused by dynamical local interactions, is not any longer applicable in quantum mechanics and that in some situations, the state of the system prevents us from attributing properties to its subsystems. Schr\"odinger called this property `entanglement' and it constitutes until today one of the most significant features of quantum mechanics.

\begin{figure}[h]
\centering
\includegraphics[width=12cm]{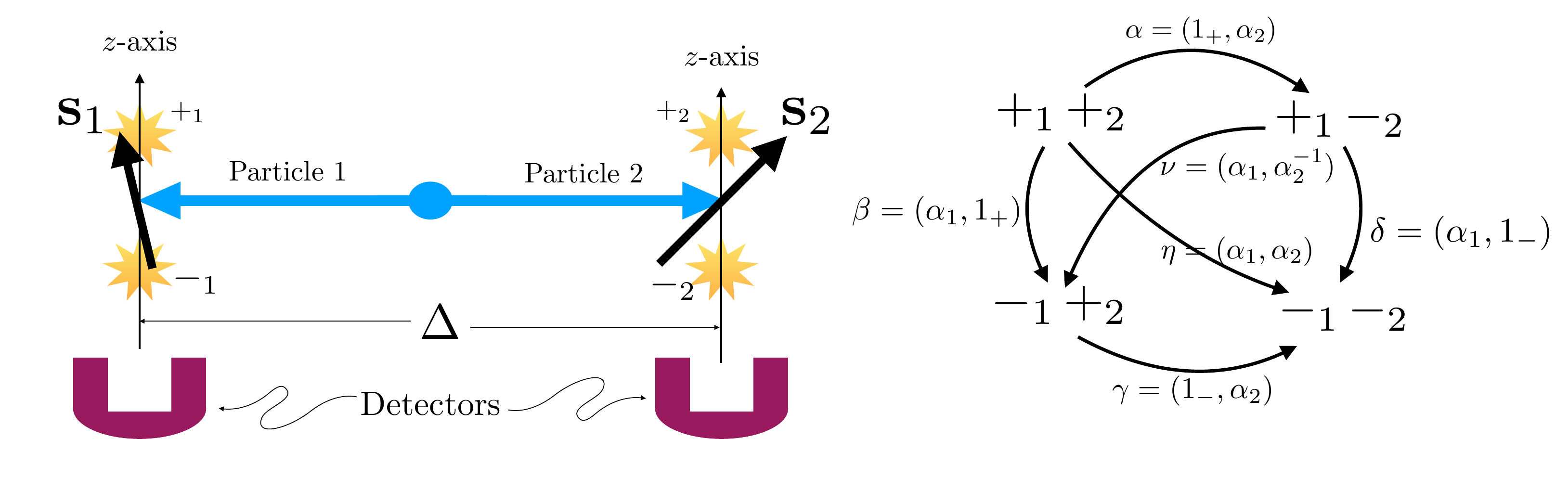} 
\caption{The EPRB experiment (left). The groupoid diagram of the EPRB experiment (the inverses $\alpha^{-1}, \beta^{-1}$, etc., and the units $1_{++},1_{+-}$, etc., are missing): $\mathbf{G}_\mathrm{EPRB} \cong \mathbf{A}_2 \times \mathbf{A}_2$ (right).}
\label{fig:eprb}
\end{figure}

In what follows we will offer an analysis of the EPRB experiment using the framework provided by the groupoid picture of quantum mechanics. One of the salient features of this analysis will consist of showing that not only the entanglement of the state describing the state of the two particles of spin 1/2 is responsible for their interdependence but that it does not make sense of talking of them as subsystems of the total system, hence of considering the total system as a direct product of individual systems.

\subsection{Groupoids and the EPRB experiment}

The groupoids picture of the EPRB system is given by the possible outcomes of the detectors of spin $z$-component.   Their outcomes are collected as pairs $(\pm_1,\pm_2)$, thus the space of outcomes $\Omega$ consists of four outcomes denoted as $+_1+_2$, $+_1-_2$, $-_1+_2$ and $-_1-_2$, respectively (see Fig. \ref{fig:eprb} right).   The groupoid picture will be completed by establishing the possible physical transitions among the outcomes, so that in the simplest possible scheme we can consider the following ones $\alpha \colon +_1+_2 \to +_1 -_2$, $\beta \colon +_1+_2 \to -_1+_2$, $\gamma \colon -_1+_2 \to -_1-_2$, $\delta \colon +_1-_2 \to -_1 -_2$, $\eta \colon +_1+_2 \to -_1-_2$, $\nu\colon +_1-_2 \to -_1+_2$, and their respective inverses together with the corresponding units $1_{++}$, $1_{+-}$, $1_{-+}$ and $1_{--}$.   Thus, for instance $\alpha^{-1} \circ \alpha = 1_{++}$, $\delta \circ \alpha = \eta$, $\nu = \alpha \circ \beta^{-1} = \gamma^{-1} \circ \delta$, etc.

Summarizing, the groupoid $\mathbf{G}_\mathrm{EPRB} \rightrightarrows \Omega$ describing the EPRB experiment in its simple idealized way will consist of 16 transitions: 
$$
\mathbf{G}_\mathrm{EPRB} = \{1_{++}, 1_{+-}, 1_{-+},1_{--}, \alpha, \beta, \gamma, \delta, \eta, \nu,  \alpha^{-1} , \beta^{-1} , \gamma^{-1} , \delta^{-1} , \eta^{-1} , \nu^{-1} \} \, ,
$$ 
with the corresponding source and target maps and the composition law derived from the diagram in Fig. \ref{fig:eprb} (right).  

The groupoid algebra $\mathbb{C}[\mathbf{G}_\mathrm{EPRB}]$ of this system will be 16 dimensional and isomorphic to the space of $4\times 4$ matrices $M_{16}(\mathbb{C})$.    One way to establish this isomorphism is by considering the fundamental representation $\pi_0 \colon \mathbb{C}[\mathbf{G}_\mathrm{EPRB}] \to \mathcal{B}(\mathcal{H}_\Omega)$, where $\mathcal{H}_\Omega$ is the fundamental Hilbert space generated by the outcomes $(\pm_1,\pm_2)$, $a = 1,2$ of dimension 4.   Denoting the vectors in such space as $|\pm_1\pm_2 \rangle$, the fundamental representation is defined as (see \cite{Ib18b} for details) as:
$\pi_0(\alpha)|+_1+_2\rangle = |+_1-_2\rangle$ and zero otherwise (similarly for all other transitions).  Denoting by $[\alpha ]$, $[ 1_{+-}]$, etc.,  the matrices associated to the corresponding transitions in the given orthonormal basis $\{ |+_1 +_2\rangle, |+_1 -_2\rangle, |-_1 +_2\rangle, |-_1 -_2\rangle \}$ we get, for instance:
$$
[\alpha ] = \left[\begin{array}{cccc} 0 & 0 & 0 & 0 \\  1 & 0 & 0 & 0 \\ 0 & 0 & 0 & 0 \\ 0 & 0 & 0 & 0\end{array} \right] \, , \qquad [1_{-+} ] = \left[\begin{array}{cccc} 0 & 0 & 0 & 0 \\  0 & 0 & 0 & 0 \\ 0 & 0 & 1 & 0 \\ 0 & 0 & 0 & 0\end{array} \right] \, , 
$$
etc.
A state $\rho \colon \mathbb{C}[\mathbf{G}_\mathrm{EPRB}] \to \mathbb{C}$ can be specified by defining its values on the elements of $\mathbf{G}_\mathrm{EPRB}$, thus for instance we may define the following state:
$$
\rho_0 (1_{+-}) = \rho_0(1_{-+}) = 1/2 \, , \qquad \rho_0(\nu) = \rho_0(\nu^{-1}) = -1/2 \, , 
$$
and zero in all other transitions, that is $\rho_0(1_{++}) = \rho_0(1_{--}) = \rho_0(\alpha) = \rho_0(\beta) = \ldots = \rho_0(\eta^{-1}) = 0$.  Notice that $\rho_0 (\mathbf{1}) = 1$, where the identity element of the groupoid algebra is $\mathbf{1} = 1_{++} + 1_{+-} + 1_{-+} + 1_{--}$.    The function $\varphi_0 \colon \mathbf{G}_\mathrm{EPRB} \to \mathbb{C}$ thus defined on the groupoid is a positive
semidefinite function characterizing the state (see \cite{Ib18c} for details).  

Two  physical interpretations of the  state $\rho_0$ can be offered.  The first. more direct one, is that the expected value of the observable defined by the element $1_{+-}$, for instance, is $1/2$; $-1$ for the observable $\boldsymbol{a} = \nu + \nu^{-1}$ (notice that $\boldsymbol{a}^* = \boldsymbol{a}$), etc.     Sorkin's quantum measure allows for another interpretation (that complements the previous one).  The state $\rho_0$, by means of its characteristic function $\varphi_0$, determines a decoherence functional $D$:
$$
D(\mathbf{A},\mathbf{B}) = \sum_{\sigma\in \mathbf{A}, \sigma'\in \mathbf{B}} \varphi_0(\sigma^{-1}\circ \sigma') \, ,
$$
and, a quantum measure $\mu(\mathbf{A} ) = D(\mathbf{A},\mathbf{A})$, for any measurable set $A \subset \mathbf{G}_\mathrm{EPR}$ \cite{Ib18c}, called `events' in Sorkin's terminology.   Thus, given an outcome, for instance $+_1+_2$, we consider the quantum measure of the event determining it, that is, all transitions ending on it, $\mathbf{A} = \{ \sigma \colon x \to +_1+
_2 \} = \{ 1_{++}, \alpha^{-1} , \beta^{-1}, \eta^{-1}\}$.   Then:
\begin{eqnarray*}
\mu (\mathbf{A}) &=& \sum_{\sigma, \sigma'\in \mathbf{A}} \varphi_0(\sigma^{-1}\circ \sigma') = \varphi_0(\alpha\circ \alpha^{-1}) + \varphi_0(\beta\circ \beta^{-1}) + \\ &+& \varphi_0(\alpha\circ \beta^{-1}) + \varphi_0(\beta \circ \alpha^{-1})  = 1/2 + 1/2 - 1/2 -1/2 = 0 \, ,
\end{eqnarray*}
where only the non-vanishing terms, i.e.,  those terms such that $\sigma^{-1}\circ \sigma' = 1_{+-}, 1_{-+}, \nu$, or $\nu^{-1}$,  have been written in the r.h.s. of the previous equation.
Hence, according to Sorkin's interpretation, the event $\mathbf{A}$, that is the firing of both detectors with spin $z$ component equal to $+1/2$, is precluded and it cannot happen (compare with the analysis of the two-slit experiment in \cite{Ib18c}).  

On the other hand, we may consider the event corresponding to the outcomes $-_1+_2$, that is, $\mathbf{B} = \{ 1_{-+}, \beta, \nu, \gamma^{-1},  \}$.  Note that this is the event indicating that the the spin $z$ component of the particle 1 is $- 1/2$ while that of particle 2 is $1/2$ respectively.  Then a simple computation shows that:
\begin{eqnarray*}
\mu (\mathbf{B}) &=& \sum_{\sigma, \sigma'\in \mathbf{B}} \varphi_0(\sigma^{-1}\circ \sigma') = \varphi_0(1_{-+})  +   \varphi_0(\nu^{-1} \circ \nu) + \\ &+& \varphi_0(1_{-+}\circ \nu) = 1/2 + 1/2  -1/2 = 1/2 \, .
\end{eqnarray*}
Similarly we may have considered the outcomes $+_1-_2$, i.e., $\mathbf{C} = \{ 1_{+-}, \alpha, \nu^{-1}, \delta^{-1}  \}$, and we get again $\mu (\mathbf{C}) = 1/2$ and the addition of both values is $\mu (\mathbf{B}) + \mu (\mathbf{C}) = 1$.

If we wish to obtain an interpretation of all previous facts in terms of familiar Hilbert spaces and vectors, we could use the GNS representation associated to the state $\rho_0$, however it is more convenient to consider the fundamental representation $\pi_0$ again, and we see that the state $\rho_0$ is the state corresponding to the vector $|\psi_\mathrm{EPRB}\rangle = \frac{1}{\sqrt{2}} \left( |+_1-_2 \rangle - |-_1+_2 \rangle  \right) \in \mathcal{H}_\Omega$, that is:
$$
\rho_0(\sigma ) = \langle\Psi_\mathrm{EPRB} |\pi_0(\sigma) |\Psi_\mathrm{EPRB} \rangle \, , \qquad \sigma \in \mathbf{G}_\mathrm{EPRB} \, ,
$$
and then we can recover the usual language of quantum mechanics with amplitudes, probabilities, etc.

\subsection{Separability of the EPRB system}

It is relevant to remark that in all previous discussion, no mention or use of the `subsystems' defined by the individual particles has been made.   It is very tempting to interpret the total space of outcomes $\Omega$ as the Cartesian product of the spaces of outcomes of the individual detectors, that is $\Omega = \Omega_1 \times \Omega_2$, $\Omega_1 = \{ +_1, -_1\}$ and $\Omega_2 = \{ +_2, -_2\}$.   In doing so, the fundamental Hilbert space $\mathcal{H}_\Omega$ can be identified with the tensor product $\mathcal{H}_{\Omega_1} \otimes \mathcal{H}_{\Omega_2}$, with a basis given by the vectors $|\pm_1 \rangle \otimes |\pm_2\rangle$, that can be identified with the vectors $|\pm_1, \pm_2 \rangle$.       It is also clear that the groupoid $\mathbf{G}_\mathrm{EPRB}
$ is isomorphic to the direct product of the groupoids of pairs $\mathbf{A}_2$ of the spaces of outcomes $\Omega_a$, $a = 1,2$.  Thus for the outcomes corresponding to the detector 1, we will have the groupoid $\mathbf{G}_1 \rightrightarrows \Omega_1$ with transitions $\alpha_1 \colon +_1 \to -_1$, $\alpha^{-1} \colon -_1 \to +_1$ and the units $1_+$, $1_-$, and accordingly for the outcomes corresponding to the dectector 2. Thus, $\mathbf{G}_\mathrm{EPRB} \cong \mathbf{G}_1 \times \mathbf{G}_2$ and we can identify the transitions in $\mathbf{G}_\mathrm{EPRB}$ with pairs of transitions in $\mathbf{G}_1$ and $\mathbf{G}_2$ respectively, thus $\alpha = (1_+, \alpha_2)$, $\beta = (\alpha_1, 1_+)$, etc. (see Fig. \ref{fig:eprb} right).  Then, we use the discussion in Sect. \ref{sec:direct} and we identify the groupoid algebra $\mathbb{C}[\mathbf{G}_\mathrm{EPRB}]$ with the tensor product of algebras $\mathbb{C}[\mathbf{A}_2]\otimes \mathbb{C}[\mathbf{A}_2] \cong M_2(\mathbb{C}) \otimes M_2(\mathbb{C})$. The fundamental representation $\pi_0$ is the tensor product of the fundamental representations of the groupoids $\mathbf{A}_2$ (which is just the standard defining representation of $M_2(\mathbb{C})$ on $\mathbb{C}^2$ and we recover the well known mathematical formalism.  

The almost inevitable risk in doing so is that such point of view seems to imply that   each particle are well defined as individual systems and that they are subsystems of the total system.  This is simply not true.   As discussed in Sect. \ref{sec:direct},  each one of the algebras of the groupoids $\mathbf{G}_a$, $a = 1,2$, can be thought of as subalgebras of the total algebra by mapping each element $\boldsymbol{a}_1 \in \mathbb{C}[\mathbf{G}_1]\mapsto \boldsymbol{a}_1 \otimes I_2 \in \mathbb{C}[\mathbf{G}_\mathrm{EPRB}]\cong M_2(\mathbb{C}) \otimes M_2(\mathbb{C})$ (and similarly for $\boldsymbol{a}_2\in \mathbb{C}[\mathbf{G}_2]$).  However these subalgebras   are not physical subsystems as they do not retain the physical properties of the original system manifested by the transitions and outcomes of the total system.  As it was discussed extensively in 
Sect. \ref{sec:subsystems}, a physical subsystem corresponds to the mathematical notion of subgroupoid and the groupoids $\mathbf{G}_a$, $a = 1,2$ are not subgroupoids of the direct product groupoid $\mathbf{G}_\mathrm{EPRB} \cong \mathbf{G}_1 \times \mathbf{G}_2$.    We can understand now this situation better in the context of our experiment.   The most natural way of selecting a subgroupoid related to particle 1 say, is by considering the subgroupoid $\widetilde{\mathbf{G}}_1 = \mathbf{G}_1 \times \boldsymbol{\Omega}_2$, i.e., the pairs $(1_+, 1_+), (\alpha_1, 1_+), (\alpha_1^{-1}, 1_+), (1_-,1_+)$ and the same with $1_-$ in the second component (note that the subgroupoid $\widetilde{\mathbf{G}}_1$ has order 8), i.e., the subgroupoid $\widetilde{\mathbf{G}}_1$ takes into account that we have a subsystem of the total system by using all possible outcomes of the second detector as an `idle' component.   The same argument applies to the second particle and its detector.

What is important to realize now is that the algebras of the subgroupoids $\widetilde{\mathbf{G}}_1$ and $\widetilde{\mathbf{G}}_2$ do not commute anymore as subalgebras of the algebra of $\mathbf{G}_\mathrm{EPRB}$ (note, for instance, that $\alpha = (\alpha_1, 1_+) \cdot (1_+,1_+ ) \neq  (1_+,1_+ ) \cdot (\alpha_1, 1_+) = 0$).  Thus the subalgebras corresponding to the two subgroupoids are not independent in the usual sense (Def. \ref{def:independence_usual}).   But, even more, they are not independent in the generalized sense either.  For instance we get: 
$$
\rho_0 ((1_+,\alpha_2^{-1})\cdot (\alpha_1, 1_+) ) = \rho_0 (\alpha^{-1}\cdot \beta ) = \rho_0 (\nu) = 1/2\, ,
$$
while $\rho_0 ((1_+,\alpha_2^{-1})) = \rho_0(\alpha^{-1}) = 0$ and $\rho_0((\alpha_1, 1_+)) = \rho_0 (\beta ) = 0$ contradicting Def. \ref{def:independence_gen}.  Thus there is not a reasonable way in which the two particles of the experiment can be considered independent systems and this impossibility of decomposing the system into two independent subsystems is previous to the entangled character of the state $\rho_0$ defining it.  In this vein we may add that the restriction of the state $\rho_0$ to the two subsystems $\widetilde{\mathbf{G}}_a$. $a = 1,2$ is well defined. They define states $\rho_a = \rho_0\mid_{\widetilde{\mathbf{G}}_a}$, whose only non-zero contributions are given by $\rho_a(1_{+-}) = \rho_a(1_{-+}) = 1/2$ and they correspond to classical states with distributions $(0,1/2,1/2, 0)$ on the space $\Omega = \{ ++, +-,-+,--\}$.

We will end this discussion by observing that even the description of the system as a direct product  $\mathbf{G}_1 \times \mathbf{G}_2$ is uncertain.   We may recall from the discussion in Sect. \ref{sec:direct} that such naive interpretation of Lieb-Yngvason definition of compound system requires that the pairs of outcomes $(x_1,x_2)$ corresponding to the two systems should be determined simultaneously which, because of the spatial separation among the particles at the moment that the analysis is performed, is not the case in our experimental setup.  We will thus conclude that any attempt to discuss the nature of the system after a spatial separation between the given particles has been reached requires further analysis (even if the initial configuration was a direct product) that should be provided involving a quantum field theoretical treatment.


\section{Conclusions and discussion}

The discussion of the concept of subsystem of a quantum system described in the groupoids picture shows that it corresponds to the notion subgroupoid.  After discussing their basic properties the notion of composition has been analyzed.  Two main composition mechanisms are explored corresponding to the standard notion of direct product and of the free product of groupoids, that extends the notion of free product of groups.    These are shown to capture the two universal properties of groupoids (physical systems) with respect to immersions (subsystems) and to projections (reduction).

Using these notions the problem of determining the independence of two subsystems has been addressed showing that the two natural notions of independence which are natural in the domain of non-commutative probability theory, can be extended nicely to the setting of quantum systems described by groupoids.   The notion of generalized independence with respect to a family of states of interest is introduced and it is shown that the direct product of groupoids equipped with the natural class of states defined on the factors are shown to define independent systems with respect to the family of separable states.

In spite of this, it is shown that the direct product of groupoids exhibits a deeper `non-separability' property and that there is not a unique and/or natural way to consider the factor groupoids as subgroupoids, or subsystems, of the direct product.  This leads to the observation that in the physical interpretation of experiments like the EPRB experiment, beyond the entangled character of the state of the system, there is an intrinsic non-separabilty property attached to the physical structure of the compound physical system.

Obviously similar considerations can be extended to other debatable aspects of the foundations of Quantum Mechanics like the interpretation of the Schr\"odinger's cat experiment.  This and other relevant experiments will be discussed in a forthcoming work.


\section*{Acknowledgements}
The authors acknowledge financial support from the Spanish Ministry of Economy and Competitiveness, through the Severo Ochoa Programme for Centres of Excellence in RD (SEV-2015/0554).
AI and FdC would like to thank partial support provided by the MINECO research project  MTM2017-84098-P  and QUITEMAD+, S2013/ICE-2801.   GM would like to thank partial financial support provided by the Santander/UC3M Excellence  Chair Program 2019/2020.


\end{document}